\documentclass[11pt]{article}

\usepackage{amstext, amsmath, amssymb, amsthm}
\usepackage{kpfonts}
\usepackage{color}
\usepackage{enumerate}
\usepackage{fullpage}  

\usepackage[margin=1in]{geometry}

\newtheorem{theorem}{Theorem}[section] 
\newtheorem{corollary}[theorem]{Corollary}
\newtheorem{lemma}[theorem]{Lemma}

\newtheorem{proposition}[theorem]{Proposition}
\newtheorem{definition}[theorem]{Definition}
\newtheorem{claim}[theorem]{Claim}

\newcommand{\val}{\mathrm{val}}
\newcommand{\eps}{\varepsilon}

\newcommand{\oline}[1]{\overline{#1}}
\newcommand{\poly}{\mathrm{poly}}

\newcommand{\scheme}{\mathcal{S}}

\newcommand{\G}{\mathcal{G}}

\DeclareMathOperator{\E}{\mathbb{E}}

\begin{document}

\title{A No-Go Theorem for Derandomized Parallel Repetition: \\ Beyond Feige-Kilian}
\author{Dana Moshkovitz 
        \thanks{
		  {\tt dmoshkov@csail.mit.edu}.} \\ MIT
\and Govind Ramnarayan \thanks{
		  {\tt govindr@csail.mit.edu}.}  \\ MIT
\and Henry Yuen  \thanks{
		  {\tt hyuen@csail.mit.edu}.} \\ MIT }
\date{}
\clearpage\maketitle

\thispagestyle{empty}

\begin{abstract}
In this work we show a barrier towards proving a randomness-efficient parallel repetition, a promising avenue for achieving many tight inapproximability results. Feige and Kilian (STOC'95) proved an impossibility result for randomness-efficient parallel repetition for two prover games with {\em small degree}, i.e., when each prover has only few possibilities for the question of the other prover. In recent years, there have been indications that randomness-efficient parallel repetition (also called \emph{derandomized parallel repetition}) might be possible for games with large degree, circumventing the impossibility result of Feige and Kilian. In particular, Dinur and Meir (CCC'11) construct games with large degree whose repetition can be derandomized using a theorem of Impagliazzo, Kabanets and Wigderson (SICOMP'12). However, obtaining derandomized parallel repetition theorems that would yield optimal inapproximability results has remained elusive.

This paper presents an explanation for the current impasse in progress, by proving a limitation on derandomized parallel repetition. We formalize two properties which we call ``fortification-friendliness'' and ``yields robust embeddings''. We show that any proof of derandomized parallel repetition achieving almost-linear blow-up cannot both (a) be fortification-friendly and (b) yield robust embeddings. Unlike Feige and Kilian, we do not require the small degree assumption. 

Given that virtually all existing proofs of parallel repetition, including the derandomized parallel repetition result of Dinur and Meir, share these two properties, our no-go theorem highlights a major barrier to achieving almost-linear derandomized parallel repetition.
\end{abstract}


%
%
%

\newcommand\remove[1]{}

\section{Introduction}

\subsection{Parallel Repetition and Almost Linear Blowup}

Two prover games are central objects of study in probabilistically checkable proofs (PCPs)~\cite{babai1990nondeterministic,Raz,khot2002power}, cryptography~\cite{ben1988multi,ben1990efficient}, and quantum computing~\cite{cleve2004consequences,reichardt2013classical}. In a two prover game $G$, two all-powerful provers coordinate their strategies and are then sent to different rooms, where they can no longer communicate. A verifier samples a pair of correlated questions $(x,y)$, and sends one question to each prover. Each prover sends back an answer, and the verifier accepts only if the pair of answers $(a,b)$ satisfy some constraint $\pi_{(x,y)}$ depending on the questions. The {\em value} of the game $G$, denoted $\val(G)$, is the probability that the verifier accepts, maximized over all prover strategies.

Parallel repetition is a natural transformation to amplify the hardness of two prover games. The \emph{$k$-fold parallel repetition} of a game $G$, denoted $G^k$, is another two prover game where the verifier picks $k$ independent pairs of questions from $G$, and sends each prover $k$ questions, corresponding to half of each of the $k$ question pairs. Each  prover sends back $k$ answers and the verifier accepts if it would have accepted all $k$ pairs of answers in the original game. Clearly, if the provers have strategies that make the verifier accept in the original game with probability $1$ (i.e., $\val(G) = 1$), then they can make the verifier accept in the $k$-fold repetition with probability $1$.
The celebrated parallel repetition theorem of Raz~\cite{Raz} shows that if the value of the game $G$ is smaller than $1$, then the value of the $k$-fold repetition, $\val(G^k)$, decays exponentially with $k$.

One of the most important applications of parallel repetition is in hardness of approximation, where it is used in reductions proving inapproximability results~\cite{Has97}. However, this application reveals a significant disadvantage of parallel repetition: the randomness complexity of the verifier in $G^k$ is $k$ times the randomness complexity of the original game $G$. This increase corresponds to a blow-up of $k$ in the exponent in reductions that are based on parallel repetition. As a result, if a reduction from \textsc{Sat} on size-$n$ inputs applies $k$-fold parallel repetition to derive an instance of a target problem, then the resulting instance of the target problem takes inputs of size $O(n^k)$. Hence, the conjectured lower bound of $2^{\Omega(n)}$ on the time needed to solve \textsc{Sat} translates at best to a time lower bound of $2^{\Omega(n^{1/k})}$ on the target problem. In applications, $k$ is often a large constant~\cite{Has97}. However, in order to obtain optimal inapproximability results for many problems,  one would like to apply parallel repetition $k$ times for all $k$'s up to $\Theta(\log n)$~\cite{BGLR,M}.


This motivates the fundamental question of whether \emph{derandomized} or \emph{randomness efficient} parallel repetition is possible: could an analogue of Raz's parallel repetition theorem hold even if the verifier does not pick $k$ question pairs independently, but rather picks $k$ \emph{correlated} question pairs? In particular, if the verifier of the original game uses $\log n$ random bits, one could hope for a verifier that uses $\log n + O(k)$ random bits to play the repeated game (as opposed to $k\log n$ random bits). If such a derandomized version of Raz's parallel repetition theorem were possible, then this would yield reductions from say, \textsc{Sat}, where a $2^{\Omega(n)}$ lower bound on \textsc{Sat} translates to a matching $2^{\widetilde{\Omega}(n)}$ lower bound on the target problem!

In~\cite{MR08}, Moshkovitz and Raz gave a hardness amplification transformation similar in spirit to parallel repetition where the transformed game uses only $(1+o(1))\log n + O(k)$ random bits. Such a blowup is referred to as ``almost linear'', and is now the gold standard for reductions. Unfortunately, the answer size in the transformation of~\cite{MR08} is exponential in $k$ rather than polynomial in $k$, and hence falls short of proving the so-called Projection Games Conjecture on optimal hardness of approximation. The parallel repetition transformation, on the other hand, gives an optimal tradeoff between the hardness of the resulting game (the \emph{soundness error}) and the answer size. This motivates the search for a derandomized parallel repetition theorem that uses $(1+o(1))\log n + O(k)$ random bits and has $O(k)$ answer bits for all $k\leq \log n$. This could prove the Projection Games Conjecture, as well prove tighter inapproximability results.

\subsection{The Feige-Kilian impossibility result}

Feige and Kilian~\cite{FK} proved an impossibility result for derandomized parallel repetition, showing that given a game $G$ satisfying two conditions called \emph{softness} and \emph{small degree}, the value of any randomness-efficient parallel repetition of $G$ is \emph{independent} of the number of repetitions. The softness condition means that if $G$ has randomness complexity $\log n$, then $G$ is \emph{$n^\eps$-soft} iff on any subset of $n^\eps$ question pairs the verifier may ask, there exists a strategy for the provers to win with probability $1$. A game has \emph{degree-$d$} if for any question of one prover, the largest number of questions for the other prover is at most $d$.
Specifically, their main result is the following:
\begin{theorem}[Feige-Kilian]
\label{thm:feige-kilian}
	Let $G$ be a two prover game with $n$ possible question pairs. If $G$ is $n^\eps$-soft and has degree $d$, then for any game $H$ that involves playing $k$ correlated instances of $G$, if the randomness complexity of the verifier of $H$ is at most $c \log n$, then the value of $H$ is independent of $k$; in particular, $\val(H) \geq (2d)^{-4 c^2/\eps^2}$.
\end{theorem}	
Here, we will call $G$ the \emph{base game} and $H$ the \emph{$k$-repeated game}.
Next we describe the argument of Feige and Kilian in the almost linear regime (i.e., the repetition only uses $(1 + \eps) \log n$ random bits). In this regime their argument takes an especially simple form: because the base game has small degree, the provers have constant probability to guess each other's question in the first round, and if they succeed, there are only $n^{\eps}$ possibilities for the rest of the $k - 1$ questions. For soft games the provers can succeed on all remaining questions -- thus  the provers' success probability in the repeated game does not decay with the number of repetitions $k$.

 The softness condition is satisfied by games of interest. If we assume that solving \textsc{Sat} requires more than $2^{n^\eps}$ time, then the games we apply parallel repetition to will be in general $n^\eps$-soft. The small degree condition -- while true of some games to which standard parallel repetition is applied -- is not necessarily satisfied by all games of interest. 
In other words, Feige and Kilian's impossibility result imposes a strong limitation on the possibility of derandomized parallel repetition when working in the ``small degree regime'' -- i.e., when the degree of $G$ is a constant independent of the randomness complexity or the number of repetitions -- but leaves the fascinating open question: can one obtain randomness-efficient parallel repetition for the ``large degree regime'', in which the degree of the game $G$ can depend on its randomness complexity or the desired number of repetitions. In particular, Feige and Kilian do not rule out degree that is inversely proportional to the desired value of the repeated game.

Indeed, a few works have explored this avenue towards derandomized parallel repetition. Shaltiel~\cite{Sh} considered the setting of games where the questions to each prover are uncorrelated (also known as \emph{free games}). Here, the degree is maximal, and Shaltiel managed to get a modest, albeit non-trivial, savings in randomness complexity in a repeated game. Dinur and Meir~\cite{DM} constructed games with ``linear structure'' -- which also have large degree -- and showed that a theorem by Impagliazzo, Kabanets and Wigderson~\cite{IKW} gives a certain randomness-efficient parallel repetition for them. Unfortunately, neither of these results imply new hardness of approximation results, since the reductions from \textsc{Sat} to both free games and games with linear structure generate games with randomness complexity or answer size that are very large compared to the size of the \textsc{Sat} formula.


\subsection{Our work}

This paper begins where Feige and Kilian left off: we show a barrier for derandomized parallel repetition in the large degree regime. One may hope for an analogue of Feige and Kilian's negative result for large degree games, but, unfortunately, this seems to be impossible. The reason is that in fact there \emph{are} games for which we can decrease error in a randomness efficient fashion, but without performing derandomized parallel repetition in a meaningful sense. Specifically, we can construct a high-error base game $G$ that actually ``hides'' a game $G_{low}$ for which we already know that $\val(G_{low}) \leq \delta$; if then we apply a derandomized parallel repetition procedure such as Dinur's graph powering~\cite{Dinur} to $G$, we obtain a repeated game $H$ that closely approximates $G_{low}$ and thus $\val(H) \lesssim \delta \ll \val(G)$. For more details, see Appendix~\ref{contrived_example}. Thus we've obtained derandomized error reduction, but intuitively the low error didn't come from the parallel repetition, but rather from the planted low-value game $G_{low}$.

This example shows that we can't hope to extend Theorem~\ref{thm:feige-kilian} directly to large degree games. Instead, we do the next best thing: we prove a limitation on \emph{proof techniques} for derandomized parallel repetition. We formalize two proof properties which we call ``fortification-friendliness'' and ''yields robust embeddings'', and then show that any proof of almost-linear derandomized parallel repetition cannot simultaneously be fortification-friendly and yield robust embeddings. Nearly all proofs of parallel repetition -- even derandomized parallel repetition theorems -- are fortification friendly and yield robust embeddings, including: Raz~\cite{Raz}, Shaltiel~\cite{Sh}, Dinur-Meir~\cite{DM}, Impagliazzo, Kabanets and Wigderson~\cite{IKW}, Moshkovitz~\cite{M14}, and Braverman-Garg~\cite{BG}. Therefore our results explain why their techniques have not been pushed to almost linear size.

We now discuss these two properties in more detail.


\subsubsection{Proof of parallel repetition by robust embedding}
The key step in proofs of the parallel repetition theorem is to argue that the success probability of the average coordinate $i$ of $G^k$ cannot be much larger than $\val(G)$, even when conditioned on the provers winning a significant fraction of coordinates that don't include $i$. This is proved via reduction: if this were not true, then the provers extract a strategy for $G$ from a strategy for $G^k$ by embedding $G$ into the $i$'th coordinate of $G^k$ conditioned on winning a set $C$ of coordinates. However, if $\val(G^n)$ is too large, then this strategy would succeed with probability better than $\val(G)$, a contradiction. We say that such an analysis of parallel repetition is by {\em embedding}. Furthermore, the embeddings given are \emph{robust}. By robust, we mean that embedding $G$ into a coordinate of $G^k$ is possible even when conditioning on winning \emph{any} not too large subset $C$ of coordinates. We will give a more detailed overview of this embedding technique in Section~\ref{s:embedding}.



\subsubsection{Fortification-friendly repetition schemes} Our no-go theorem covers derandomized parallel repetition theorems that can be applied to at least one \emph{fortified} game. In this case we say that the parallel repetition theorem is \emph{fortification-friendly}. Currently, there is no parallel repetition scheme that utilizes the fact that the base game is \emph{not} fortified, and hence all existing parallel repetition schemes are fortification friendly. This includes the scheme of Dinur and Meir, which we elaborate on at the end of this subsection.

Fortification is a property of games introduced in ~\cite{M14}. Roughly speaking, a $(\delta,\varepsilon)$-fortified game $G$ is one where the value of so-called ``rectangular'' subgames of $G$ that contain at least $\delta$ fraction of the questions is the same as the value of $G$ up to an additive $\varepsilon$. The paper~\cite{M14} gives a simple analysis for parallel repetition of fortified games, and furthermore showed that arbitrary games can be easily fortified by composing them with expanders.
While fortified games were defined fairly recently, they are quite natural, and, in particular, most games are fortified: see Appendix~\ref{sec:random_games_fortified}.

Importantly, existing derandomized parallel repetition theorems are fortification-friendly: Shaltiel proves a derandomized parallel
repetition for free games, which also works for fortified free games. In~\cite{DM}, Dinur and Meir first present a ``linearization'' operation that converts any game into a game with linear structure, and then prove a derandomized parallel repetition that works for \emph{any} game with linear structure. The core of this derandomized parallel repetition is the work of Impagliazzo-Kabanets-Wigderson, and the underlying derandomized parallel repetition theorem of~\cite{IKW} \emph{is} fortification friendly. This is because the result of Impagliazzo-Kabanets-Wigderson applies to all free games: (1) free games trivially have linear structure (since all possible edges are present) and (2) it is easy to construct fortified free games (e.g. choosing random constraints for a free game). Thus, our results imply limitations on what is achievable by the Impagliazzo-Kabanets-Wigderson derandomized parallel repetition, and hence what is achievable by the Dinur-Meir result.

%
\subsubsection{Informal Theorem Statement and Discussion}

We are now ready to state our main theorem informally. For a formal statement, see Theorem~\ref{main_thm}.

\begin{theorem}[Main theorem, informal statement]
\label{main_thm_informal}
Let $\scheme$ be a parallel repetition scheme that transforms any \emph{base game} $G$ to a $k$-repeated game $\scheme(G)$ in which a verifier asks $k$ (possibly correlated) questions from $G$ in parallel. Suppose that $G$ is $(\delta,\varepsilon)$-fortified for sufficiently small\begin{footnote}{The required $\delta$ depends on the blowup in $\scheme(G)$. For $|\scheme(G)| = |G|\cdot\poly\log|G|$, we need $\delta = 1/\poly\log|G|$. For $|\scheme(G)| = O(|G|)$, we need a sufficiently small constant $\delta$.}\end{footnote} $\delta= |G|^{-o(1)}$, and $\varepsilon = O(1-\val(G))$; and that $|\scheme(G)| = |G|^{1+o(1)}$.
Then there is no proof of parallel repetition by robust embedding for $\scheme(G)$.
\end{theorem}

Note that unlike the result of Feige and Kilian, our impossibility result is not limited to small degree games. In fact, fortification typically involves composing the game with a degree-$O(1/\delta)$ expander, thereby making the degree of the base game large.

\subsection{The way forward}

Despite many years of research on the subject of derandomized parallel repetition, obtaining a parallel repetition with both an exponential decay of the error and almost-linear blowup has resisted attack. The work of Dinur and Meir makes partial progress towards this goal, but -- not only it admits polynomial decay of the error and a large polynomial blowup -- it also goes through a costly ``linearization'' operation that deteriorates the parameters of the game, so it does not achieve any new results for PCP. 

We view our theorem as an explanation for the lack of progress towards the goal of derandomized parallel repetition. It shows that any proof of a derandomized parallel repetition theorem must do at least one of the following: (1) Use that the base game is \emph{not} fortified; (2) Not yield a robust embedding; and/or (3) Have a large polynomial blowup. As discussed earlier, virtually all proofs of parallel repetition do not satisfy (1) and (2). We now discuss prospects for being able to achieve (1), (2), or (3).




\paragraph{Using that the base game is \emph{not} fortified.} Is it too restricting to require that the scheme accepts a base game is fortified? We believe not, there are no known parallel repetition techniques that take advantage of the base game not being fortified. Intuitively, it seems unlikely that such a technique would help with derandomized parallel repetition, since fortification is known to facilitate parallel repetition, and composition with expanders is intuitively useful for derandomization.

\paragraph{Circumventing robust embeddings.} Again, proving parallel repetition via robust embeddings (either explicitly or implicitly) is a ubiquitous strategy. Interestingly, one approach that does not fall into the embedding framework is the randomness-efficient amplification of Moshkovitz and Raz~\cite{MR08}. They construct codes with local testers/decoders {\em that have low error}, and incorporate randomness efficient sequential repetition on the decoded symbols. 
Their technique is based on an algebraic construction of codes and the error it obtains, while low, is not low enough to prove the Projection Games Conjecture. Decreasing the error of local testers/decoders does not seem any easier than randomness-efficient error reduction for games.

\paragraph{Polynomial blowup.} Finally, our impossibility result pertains to repetitions with almost linear blowup. As we mentioned, such a blowup is currently the gold standard in PCP, and larger blowups correspond to weaker inapproximability results. Nonetheless, both the results of Shaltiel and Dinur-Meir have larger blowups. Shaltiel has a blowup that is not much smaller than standard parallel repetition, and Dinur-Meir have a polynomial blowup.

\section{Games and parallel repetition schemes}

We will use the notation $\oline{x}$ to denote tuples $(x_1,\ldots,x_k)$. For convenience of notation, we will call two sets $\varepsilon$-close if the uniform distributions on these sets are $\varepsilon$-close in total variation distance.

\paragraph{Games and strategies.} A \emph{two-prover one-round game} $G$ is specified by a tuple $(X,Y,E,\pi, \Sigma)$ where $X \times Y$ is the vertex set of a bipartite graph with edge set $E \subseteq X \times Y$, $\pi$ is a set of constraints $\pi_e \subseteq \Sigma \times \Sigma$ for each edge $e \in E$, and $\Sigma$ is a finite alphabet. The \emph{value} of a game $G$ is defined as
$$
	\val(G) := \max_{\psi_X,\psi_Y} \Pr_{(x,y) \in E} \left [ (\psi_X(x), \psi_Y(y)) \in \pi_{(x,y)} \right]
$$
where the maximum is taken over all functions $\psi_X: X \to \Sigma$ and $\psi_Y: Y \to \Sigma$, and the probability is over a uniformly random edge in $E$. We will use caligraphic $\mathcal{G}$ to denote the \emph{graph underlying $G$}, which is the bipartite graph $(X,Y,E)$. The \emph{size} of a game $G$, which we will denote by $|G|$, is defined to be the number of edges $|E|$. For a pair of maps $\psi_X : X \to \Sigma$ and $\psi_Y: Y \to \Sigma$, we call $\psi = (\psi_X,\psi_Y)$ a \emph{strategy} for $G$. For $(x,y) \in X \times Y$, we will write $\psi(x,y)$ to denote the pair $(\psi_X(x),\psi_Y(y))$. If the maximum degree of $\G$ is $d$, then we say that $G$ is a \emph{degree-$d$ game}.

\paragraph{$k$-fold parallel repetition.} The \emph{$k$-fold parallel repetition} of a game $G$ is a new game $G^k = (X^k, Y^k, E^k, \pi^k, \Sigma^k)$, where $X^k$, $Y^k$, $E^k$, and $\Sigma^k$ denote the $k$-fold Cartesian products of $X$, $Y$, $E$ and $\Sigma$ respectively, and $\pi^k$ denotes the set of constraints $\pi_{e_1} \times \pi_{e_2} \times \cdots \times \pi_{e_k}$ for every $\oline{e} = (e_1,\ldots,e_k) \in E^k$. Intuitively, in $G^k$, the verifier will sample $e_1 = (x_1,y_1), \ldots,e_k = (x_k,y_k)$ uniformly and independently at random from $E$, and send $\oline{x} = (x_1,\ldots,x_k)$ and $\oline{y} = (y_1,\ldots,y_k)$ to the first and second prover, respectively. The provers win the repeated game $G^k$ if they win $G$ in all rounds -- i.e.,~provide answers $(a_1,\ldots,a_k)$ and $(b_1,\ldots,b_k)$ from $\Sigma^k$ such that for $i \in [k]$, $(a_i,b_i) \in \pi_{e_i}$. 

\paragraph{Subgames.} Let $G = (X,Y,E,\pi,\Sigma)$ be a game. Then we say a game $G' = (X',Y',E',\pi',\Sigma')$ is a \emph{subgame of $G$} if $X' \subseteq X$, $Y' \subseteq Y$, $E' \subseteq E \cap (X' \times Y')$, $\pi' = \{ \pi_e : e \in E', \pi_e \in \pi \}$, and $\Sigma' = \Sigma$; we denote this by $G' \subseteq G$. For a subset $E' \subseteq E$, we will let $G_{E'} = (X,Y,E',\pi,\Sigma)$ denote the \emph{subgame of $G$ induced by $E'$}. Notice that the question set, constraints and alphabet of a subgame induced by a set of edges are the same as that of the original game. The only difference is that, in the subgame, we only select a subset of the question pairs that the verifier can ask, and the constraints are induced by the subset of questions.

For convenience, when the game $G = (X,Y,E,\pi,\Sigma)$ is understood from context, we will treat $G$ as the set of edges $E$; e.g., we will write $(x,y) \in G$ to denote $(x,y) \in E$.

\paragraph{Parallel repetition schemes} Let $G = (X,Y,E,\pi,\Sigma)$ be a game, and let $k > 0$ be an integer. Then we say any subgame $H = (X^k,Y^k,E_H,\pi^k,\Sigma^k) \subseteq G^k$ where $E_H \subseteq E^k$ is a \emph{$k$-repeated game}, with $G$ as the \emph{base game}. If $|H|$ is strictly smaller than $|G|^k$, then we say that $H$ is a \emph{derandomized} $k$-repeated game.

A $k$-\emph{parallel repetition scheme} $\scheme$ is a black box procedure for converting a base game $G$ to a $k$-repeated game $H \subseteq G^k$. In this paper, we will use the shorthand $\scheme = \{ G \to H \subseteq G^k \}$ to succinctly describe the scheme $\scheme$, where we implicitly assume the transformation  $G \to H$ is described by an algorithm that runs in time polynomial in the description of the input game, as well as $k$. Whenever a parallel repetition scheme (or simply a \emph{repetition scheme}) $\scheme$ is understood from context, $H$ will always refer to the $k$-repeated game that is the scheme $\scheme$ applied to some base game $G$. We will also use $\scheme(G)$ to denote the scheme applied to $G$.

We say that a $k$-parallel repetition scheme $\scheme = \{ G \to H \subseteq G^k \}$ satisfies the \emph{uniform marginals property} if for all games $G = (X,Y,E, \pi,\Sigma)$, the marginal distribution of questions sampled from $H = (X^k,Y^k,E_H,\pi^k,\Sigma^k) = \scheme(G)$ in any single coordinate is the same as the distribution of questions in $G$. Namely, for all coordinates $j \in [k]$ and any fixed edge $(x,y) \in E$, we have that
	\[ \Pr_{(\oline{x},\oline{y}) \in E_H} \left[ (\oline{x}_j, \oline{y}_j) = (x,y) \right] = \frac{1}{|E|}. \]
	
	The uniform marginals property is an extremely mild and natural condition, which holds for all existing parallel repetition schemes. In fact, this condition even seems morally necessary for parallel repetition, as it says that each coordinate of the repeated game $H$ should look like the base game $G$, which is what we expect of a repeated game.
	
Finally, we define the \emph{size blowup} of a scheme $\scheme$ to be a function $\Phi_{\scheme,\Sigma}: \mathbb{N} \to \mathbb{R}$ defined as
$$
	\Phi_{\scheme,\Sigma} (n) := \max_{G : |G| = n} \frac{|H|}{|G|}
$$

where the maximum is over all base games $G$ with $n$ question pairs and answer alphabet $\Sigma$, and $H$ denotes the scheme $\scheme$ applied to $G$. Note that the number of games with $n$ question pairs and answer alphabet $\Sigma$ is finite\footnote{We assume that in the transformation from base game $G$ to $k$-repeated game $H$, the scheme does not care about the actual labels of the questions, and what only matters are the correlations between questions, as captured by the edge set $E$ of the base game $G$. This is consistent with existing parallel repetition schemes.}. The size blowup of a scheme captures the blowup in randomness complexity in the following way: if the base game $G$ has randomness complexity $\log n$ and the $k$-repeated game $H$ has randomness complexity at most $\log n + \ell(n)$, then the size blowup $\Phi_{\scheme,\Sigma} (n) \leq 2^{\ell(n)}$.

\paragraph{Winning in a set of coordinates.} For any $k$-repeated game $H$, any strategy $\psi$ for $H$, and any subset of coordinates $C \subseteq [k]$, let $W_C^{\psi}$ denote the subgame of $H$ consisting of all question pairs $(\oline{x},\oline{y}) \in H$ such that $(\oline{a},\oline{b}) = \psi(\oline{x},\oline{y})$ satisfies $(\oline{a}_i,\oline{b}_i) \in \pi_{\oline{x}_i,\oline{y}_i}$ for all $i \in C$. In other words,$W_C^{\psi}$ is the set of all question pairs in $H$ where the strategy $\psi$ is able to succeed in all the coordinates of $C$. We call $W_C^{\psi}$ the \emph{subgame where $\psi$ wins in $C$}. When the strategy is $\psi$ is understood from context, we will omit $\psi$ and simply write the subgame as $W_C$.

%
%

\section{Parallel repetition via embeddings}\label{s:embedding}

\newcommand{\emb}{{\sf Emb}}

In this section, we formalize the notion of an embedding as described in the introduction and expand on how it is used to prove parallel repetition. First, we will motivate the idea of embedding by giving a high level and informal discussion of proofs of parallel repetition. Then, we will formally define the notion of a \emph{robust embedding} that we will use in this paper. The idea of robust embeddings is implicit in nearly all proofs of parallel repetition: to illustrate this, we show how it is implicit in the Raz-Holenstein proof in Appendix A.

As alluded to in the introduction, most proofs of parallel repetition proceed via \emph{reduction}: the value of the repeated game $G^k$ is related to the value of the base game $G$ by exhibiting a transformation that takes a ``too good'' strategy for the repeated game $G^k$ and constructs a ``too good'' strategy for the base game $G$. Furthermore, this transformation is black box, in the sense that it works for arbitrary games $G$ and their parallel repetitions. 

How might such a generic transformation work? Intuitively, it seems that one must have a generic way of identifying a substructure within a hypothetical too-good-to-be-true strategy $\psi$ for the repeated game $G^k$, a strategy $\varphi$ for the base game $G$ that succeeds with too-high probability (i.e., strictly greater than $\val(G)$, which would be a contradiction). Since our only constraint on $G^k$ is that it's comprised of $k$ independent copies of $G$, it seems that we have to identify a strategy for $G$ within substructures of $\psi$ that respect this constraint. 

Under $\psi$, we have that $\Pr[W_{[k]}]$, the probability of winning all rounds, is too large. Thus, we can use Bayes' rule to split it into conditional probabilities that respect the coordinate structure of $G^k$. It is not hard to see that, assuming $\Pr[W_{[k]}]$ is too large, then there exists a set of coordinates $C \subset [k]$ such that for many rounds $i \in [k] \backslash C$, we have that $\Pr[W_{\{i\}} | W_C] \gg \val(G)$, where $W_{\{i\}}$ denotes the event of winning round $i$ and $W_C$ denotes the event of winning all the rounds in $C$. Thus for each such $i$ it appears that we have identified candidate substructures inside $\psi$ (namely, the event $W_C$) within which we hope to extract a too-good-to-be-true strategy for $G$ (namely, by using a strategy for the $i$th round within the event $W_C$). Thus, we would like to ``play'' a copy of $G$ in the $i$th round of $W_C$, and obtain success probability that is close to $\Pr[W_{\{i\}} | W_C]$, which would be too good to be true. The constructed strategy $\varphi$ will attempt to ``play'', or \emph{embed}, the questions of $G$ into the $i$'th round of $W_C$ (which we also think of as a subgame of $G^k$). 

We call this natural proof strategy  a \emph{proof of parallel repetition by robust embedding}. This proof strategy forms the basis of most parallel repetition proofs, including existing proofs of derandomized parallel repetition, and one might expect that future derandomized parallel repetition theorems might be proved along these lines. We formalize this notion by defining the property of having a \emph{robust embedding} of a game $G$ into a $k$-repeated game $H$.


Let $G = (X_G,Y_G,E_G,\pi_G,\Sigma_G)$ and $H = (X_H,Y_H,E_H,\pi_H,\Sigma_H)$ be games. We say the map $\emb: X_G \times Y_G \to X_H \times Y_H$ is an \emph{embedding map from $G$ to $H$} (or simply an \emph{embedding map}) iff there exist maps $f: X_G \to X_H$ and $g: Y_G \to Y_H$ such that for all $(x,y) \in X_G \times Y_G$ we have $\emb(x,y) = (f(x),g(y))$. 

\begin{definition}[Robust embedding into a repeated game]
\label{def:robust_reduction} 
Let $G = (X,Y,E,\pi,\Sigma)$ be a game and let $H \subseteq G^k$ be a $k$-repeated game. Let $\gamma: [k] \to \mathbb{R}$ be a function. We say that $G$ has a $(\gamma,\eps)$-\emph{robust embedding into a coordinate of $H$} iff for all strategies $\psi_H$ for $H$ and subsets $C \subseteq [k]$, there exists an  $i \in [k] \backslash C$, there exists an embedding map $\emb: X \times Y \to X^k \times Y^k$ such that
\begin{enumerate}
	\item (\textbf{Coordinate embedding}) For all $(x,y) \in X \times Y$, we have that $(\oline{x},\oline{y}) = \emb(x,y)$ satisfies $\oline{x}_i = x$ and $\oline{y}_i = y$.
	\item (\textbf{Robustness}) If $\Pr_{\oline{e} \in E_H} [ \oline{e} \in W_C] \geq \gamma(|C|)$, then $\Pr_{e \in E} [ \emb(e) \in W_C ] \geq 1 - \eps$.
\end{enumerate}
where $E_H$ denotes the questions in $H$, and $W_C$ denotes the subgame of $H$ where $\psi_H$ wins in $C$.
\end{definition}

We use the term ``robust'' because there is an embedding from $G$ into $W_C$ for \emph{every} $C$ such that $\Pr[W_C]$ is sufficiently large. This is reminiscent of the robustness properties of pseudorandom objects such as expanders or extractors, where we have guarantees for \emph{every} sufficiently large subset of a graph (in the case of expanders) or distribution with sufficiently large min-entropy (in the case of extractors). \emph{A priori}, $G$ may not have a robust embedding into a repeated game $H \subseteq G^k$ because there may exist large $W_C \subseteq H$ that, intuitively, does not contain a copy of $G$. 

Our definition of robust embedding is heavily inspired by the Raz-Holenstein proof of the parallel repetition theorem. In Appendix~\ref{sec:raz-hol}, we explicitly describe how the Raz-Holenstein proof directly implies the existence of a robust embedding of $G$ into a coordinate of $G^k$. 

Although the main result of our paper does not unconditionally rule out derandomized parallel repetition, we do the next best thing: we rule out a particular \emph{proof technique} for proving derandomized parallel repetition, and in fact, a very natural one.

\newcommand{\ensemble}{\mathcal{E}}

\section{Our no-go theorem}


Our main theorem is the following:

\begin{theorem}[Main Theorem]
\label{main_thm}
Let $\Sigma$ be a finite alphabet. Let $\scheme = \{G \to H \subseteq G^k \}$ be a parallel repetition scheme that satisfies the uniform marginals property and has size blowup $\Phi_{\scheme, \Sigma}(n) \leq O(n^{0.49})$. Then for all $n > 0$, $\eps \in (0, 1/23)$, $\delta \leq (16 \Phi_{\scheme,\Sigma}(n) \log^2 (\Phi_{\scheme,\Sigma}(n) ))^{-1}$, an integer $d$, and for all games $G$ satisfying:
	\begin{enumerate}
		\item The graph $\mathcal{G} = (X \times Y,E)$ underlying $G$ is $d$-regular, and has at most $\eps |E|$ parallel edges.
		\item For all $S \subseteq X,T \subseteq Y$ with $|S| \geq \delta |X|,|T| \geq \delta |Y|$, we have
		$$
			\left | \frac{| E \cap (S \times T) |}{|S| |T|} - \frac{d}{|Y|} \right | \leq  \eps \frac{d}{|Y|}.	
		$$
		\item $\mathrm{val}(G) \leq 1 - 20\eps$.
		\item $G$ is $(\delta,\eps)$-fortified.
	\end{enumerate}
there does not exist a  $(\gamma,\eps)$-robust embedding of $G$ into a coordinate of $H$ for all $\gamma \leq \val(G)$, $\eps < (1 - \gamma)/23$, where $H = \scheme(G)$.
\end{theorem}


The most significant implication of our Main Theorem is that a parallel repetition scheme satisfying the uniform marginals property that (1) can be applied to a single game $G$ with the above properties and (2) yields a robust embedding cannot achieve almost linear blowup. As we elaborate below, by far the most pertinent property of $G$ is that it is sufficiently \emph{fortified}. Hence, a parallel repetition scheme attempting to achieve almost linear blowup should explicitly take advantage of the fact that its input is not fortified. Below, we discuss the properties we require of $G$, and why games that satisfy these properties are quite natural, which makes this barrier nontrivial to overcome.

 Fortification is a property of games introduced by Moshkovitz~\cite{M14}, who gave a simple proof that fortified games satisfy parallel repetition, and furthermore showed that arbitrary (projection) games can be easily fortified. Roughly speaking, a fortified game $G$ is one where the value of every not-too-small \emph{rectangular subgame} of $G$ (i.e., a subgame of $G$ played on a subgraph of $\mathcal{G}$ induced by a set of vertices). More formally:
\begin{definition}[Fortified Games]
\label{def:fortified_games}
Let $G = (X,Y,E,\pi,\Sigma)$ be a game. We say that $G$ is \emph{$(\delta, \eps)$-fortified} iff for all $S \subseteq X$, $T \subseteq Y$ with $|S| \geq \delta |X|$ and $|T| \geq \delta |Y|$, we have that
$$
	\val(G_{S \times T}) \leq \val(G) + \eps
$$
where $G_{S \times T} \subseteq G$ denotes the rectangular subgame of $G$ on the subgraph induced by the vertex set $S \times T$.
\end{definition} 
One might be cautious that we require $G$ to be $(\delta, \eps)$-fortified with potentially polynomially small $\delta$. However, many natural games on which parallel repetition theorems apply \emph{are} this fortified. For example, we demonstrate in Appendix~\ref{sec:random_games_fortified} that games that are randomly sampled from a distribution of games on regular bipartite graphs are heavily fortified. There are even heavily fortified games with linear structure: it follows from Claim~\ref{random_makes_fortified} in Appendix~\ref{sec:random_games_fortified} that by taking a free game and randomly sampling constraints, we get a game that satisfies all the requirements of Theorem~\ref{main_thm}. Since all free games have linear structure, this also has linear structure. Hence, our barrier even applies to the derandomized parallel repetition theorem of Dinur and Meir.

Another property that we require of the game $G$ is that it is a $d$-regular, bipartite expander. However, this is not restrictive as an additional property: all known constructions of fortified games satisfy the expansion condition we desire. This includes those in~\cite{M14} and~\cite{BSVV15}, which create fortified games by composing games with good expanders, as well as the random games we construct in Appendix~\ref{sec:random_games_fortified}, which are naturally on expanders. The last property that we use is that the graph underlying $G$ does not have too many parallel edges, which is a property satisfied by virtually all games researchers consider in hardness of approximation. Finally, the reason we limit the blowup of the parallel repetition scheme $\scheme$ to be at most $O(n^{0.49})$ is because we take $\delta$ to be inversely proportional to the blowup. Since $n$ is the number of \emph{edges} in the underlying graph, taking $\delta = O(n^{-0.5})$ already makes $\delta |X|$ and $\delta|Y|$ smaller than 1 for free games, and hence $(\delta, \eps)$-fortification simply does not make sense. Furthermore, we believe the most interesting case of our theorem occurs when the blowup is just $O(n^{o(1)})$, in which case our fortification constraints are relatively mild.

We now give an intuitive overview of our argument. Let $\scheme$ be the parallel repetition scheme from the theorem statement satisfying the uniform marginals property and the small size blowup condition. We will show that for games $G$ satisfying the requisite properties, the presence of a robust embedding lets us obtain a contradiction.

Given such a game $G$ and a supposed randomness-efficient parallel repetition $H \subseteq G^k$ from the scheme $\scheme$, we rule out the existence of a robust embedding of $G$ into $H$. We prove this via contradiction: if there were a robust embedding, then from the embedding we would be able to extract an assignment for $G$ that has success probability significantly greater than $\val(G)$ on some rectangular subgame of $G$. Furthermore, the fact that $H$ is not much larger than $G$ allows us to conclude that this rectangular subgame is not too small. However, this contradicts the fortification property of $G$, which states that all not-too-small rectangular subgames of $G$ have value that's not much larger than $\val(G)$. Thus no such robust embedding can exist.


We give more details about how we extract an assignment from a robust embedding. Recall that a robust embedding of $G$ to $H$ allows us to choose a subset of coordinates $C \subseteq [k]$ and a strategy $\psi^H$ for the repeated game such that, if under $\psi^H$ the probability of success in the coordinates $C$ is greater than some threshold $\gamma$ (which depends on the size of $C$), then there exists an embedding $\emb$ that maps $G$ into the subgame $W_C$ of $H$ where the provers win in $C$. 

We exploit this by letting $C$ be a singleton round $\{s\}$, and letting $\psi^H$ be a trivial strategy where the provers play optimally in round $s$, and all other rounds independently.\footnote{One might find it suspicious that we're deriving a robust embedding from such a trivial strategy, whereas in the proof of Raz's parallel repetition, for example, a robust embedding is derived from ``too-good-to-be-true'' strategies. However, one can see from Appendix~\ref{sec:raz-hol} that the Raz-Holenstein proof of parallel repetition does indeed give us a robust embedding into the subgame $W_{\{s\}}$ from this trivial strategy. Furthermore, as described in Section~\ref{s:embedding} and Appendix~\ref{sec:raz-hol}, a robust embedding is necessary but not sufficient for proving parallel repetition, which is why the robust embedding derived from the trivial strategy won't contradict the fact that the success probability for this strategy is less than $\val(G)^k$. } The probability of succeeding in round $s$ under this strategy is precisely $\val(G)$, which is larger than the threshold $\gamma$. Therefore we obtain an embedding $\emb$ from $G$ into the subgame $W_{\{s\}} \subseteq H$ where the provers win in round $s$. 

Then, we use the fact that $H$ is a randomness-efficient parallel repetition of $G$ and the uniform marginals property to conclude that over the question pairs $(x,y)$ in the base game $G$, the projection of $\emb(x,y)$ (which are question pairs in the repeated game $H$) onto round $s$ must contain a rectangular subgame $G_{M \times N}$ of $G$ that has substantial size. Since $G$ is (approximately) embedded into $W_{\{s\}}$, by definition, $\psi^H$ yields an (almost-)satisfying assignment for $G_{M \times N}$. As stated previously, this would violate the fortification property of $G$.

Though the intuition is rather straightforward, much of the proof involves dealing with the fact that $G$ doesn't \emph{perfectly} embed into $W_{\{s\}}$, but only approximately so, which introduces errors in extracting a nearly satisfying assignment for a rectangular subgame of $G$. We defer the full proof of Theorem~\ref{main_thm} to Appendix~\ref{sec:main-thm-proof}.

\section{Conclusion and Open Problems}

We show limitations on a prevalent proof strategy for derandomized parallel repetition. Specifically, we prove that any parallel repetition scheme that can be applied to a fortified game and yields a ``robust embedding'' cannot achieve almost-linear blowup.
We leave it as an open problem to extend our limitation to schemes with larger blowup. An intriguing related question is whether one can extend our results to provide limitations on derandomized parallel repetition schemes with polynomial blow-up and exponential soundness decay. This would not contradict existing results: Shaltiel's repetition has exponential soundness decay but has nearly-exponential blowup, and Dinur-Meir achieve polynomial blow-up but have polynomial soundness decay. 
As we discussed in the Introduction, the limitation of Feige-Kilian is simple in the case of almost-linear blowup, 
whereas the case of large blowup is considerably more complicated, and it is possible that extending our result to large polynomial blowup will be similarly difficult.

Our analysis takes a robust embedding and extracts from it fairly large rectangles that are nearly satisfied.
The limitation follows from providing fortified games, which do not have such rectangles, as input.
An intriguing possibility given this state of affairs is the following: Is there a technique for parallel repetition that explicitly makes use of lack of fortification in the input? Such a technique would be able to circumvent our limitation if it were applicable to derandomized parallel repetition.

A direction for amplifying two prover games that is not captured by our limitation is amplification via locally decode or reject codes~\cite{MR08}. 
These are efficient encodings with a two query tester/decoder. 
The tester/decoder is able to decode $k$-tuples of symbols from its message, or identify a corruption in the word.
One can encode the answers of the players via such a code, and then ask each prover a different query of the tester/decoder.
Whenever the tester/decoder is correct, one can simulate a randomness-efficient sequential repetition of the base game.
There are constructions of locally decode or reject codes based on low degree polynomials (See \cite{MR08} and many previous works), 
or based on direct product testing (See~\cite{IKW,DM} and many previous works).
The value of the amplified game is typically inherited from the error probability of the local tester/decoder. 
It remains an open problem to find locally decode or reject codes with substantially lower error than existing constructions.

\paragraph{Acknowledgments.} We thank Pritish Kamath and Irit Dinur for helpful discussions. We also thank the anonymous reviewers for their suggestions on the initial manuscript. This paper is based upon work supported by the National Science Foundation under grants number 1218547 and 1452302. The last author is additionally supported by Simons Foundation Fellowship (grant \#360893).

\bibliographystyle{plain}
\bibliography{citationsNov2}

\appendix

\section{Proof of Theorem \ref{main_thm} (Main Theorem)}
\label{sec:main-thm-proof}

In this section, we prove the main theorem of this paper, which we restate here for clarity:
\begin{theorem}[Main Theorem]
Let $\Sigma$ be a finite alphabet. Let $\scheme = \{G \to H \subseteq G^k \}$ be a parallel repetition scheme that satisfies the uniform marginals property and has size blowup $\Phi_{\scheme, \Sigma}(n) \leq O(n^{0.49})$. Then for all $n > 0$, $\eps \in (0, 1/23)$, $\delta \leq (16 \Phi_{\scheme,\Sigma}(n) \log^2 (\Phi_{\scheme,\Sigma}(n) ))^{-1}$, an integer $d$, and for all games $G$ satisfying:
	\begin{enumerate}
		\item The graph $\mathcal{G} = (X \times Y,E)$ underlying $G$ is $d$-regular, and has at most $\eps |E|$ parallel edges.
		\item For all $S \subseteq X,T \subseteq Y$ with $|S| \geq \delta |X|,|T| \geq \delta |Y|$, we have
		$$
			\left | \frac{| E \cap (S \times T) |}{|S| |T|} - \frac{d}{|Y|} \right | \leq  \eps \frac{d}{|Y|}.	
		$$
		\item $\mathrm{val}(G) \leq 1 - 20\eps$.
		\item $G$ is $(\delta,\eps)$-fortified.
	\end{enumerate}
there does not exist a  $(\gamma,\eps)$-robust embedding of $G$ into a coordinate of $H$ for all $\gamma \leq \val(G)$, $\eps < (1 - \gamma)/23$, where $H = \scheme(G)$.
\end{theorem}

Let $\scheme$, $n$, $\Sigma$, $\eps$, $\delta$, and $d$ be as in the theorem statement. In the proof we let $d_X$ and $d_Y$ denote the left-degree and right-degree of $\mathcal{G}$, and so $d_X = d_Y = d$, since $\mathcal{G}$ is $d$-regular. 

Let $H = (X^k,Y^k,E_H,\pi^k,\Sigma^k)$ be the $k$-repetition of $G$ under the scheme $\scheme$. Define $z := \frac{|H|}{|G|}$. Notice that $z \leq \Phi_{\scheme, \Sigma}(n)$, since $\Phi_{\scheme, \Sigma}(n)$ is effectively a maximum of $z$ taken over all games $G$ with $|G| = n$. Fix $\gamma \leq \val(G) \leq 1 - 20\eps$. Fix a round $s \in [k]$. Let $\psi^H = (\psi^H_X,\psi^H_Y)$ be a strategy for the provers in the repeated game under which the event of winning in round $s$ occurs with probability at least $\gamma$. Note that such a strategy always exists: the two provers can play the optimal strategy for $G$ in round $s$, and by the uniform marginals property of Theorem \ref{main_thm}, $\Pr[W_{\{s\}}] = \val(G) \geq \gamma$. 

Suppose for contradiction that there exists an $(\gamma,\eps)$-robust embedding into a coordinate of $H$. Let $C = \{s\}$ for some round $s \in [k]$. Then by definition of a robust embedding, we obtain an embedding map $\emb(x,y) = (f_X(x),f_Y(y))$ for maps $f_X: X \to X_H$ and $f_Y: Y \to Y_H$. Denote by $W_{\{s\}}$ the set of edges that win in round $s$ under the strategy $\psi^H$.


Define $W$ to be the set of edges that win in round $s$ and are mapped into by the embedding map, namely, $W := \{ \emb(x,y): (x,y) \in E \} \cap W_{\{s\}}$, and let $H_W$ be the subgame of $H$ induced by the edge set $W$.

Combining the fact that $\emb$ is a robust embedding into $\{s\}$ (Definition \ref{def:robust_reduction}, Property 2) with the definition of $W$, we know that 
\begin{equation}
\label{winning_edges}
 \Pr_{(x,y) \in E} [\emb(x,y) \in W] \geq 1 - \eps
 \end{equation}
 While all but an $\eps$-fraction of edges $(x,y) \in E$ map into $W$ under $\emb$, it will be convenient for us to define a set $\widehat{W}$ which \emph{all} the edges $(x,y)$ map into. Hence, we define the set of repeated game vertex pairs $\widehat{W}$ to be
 \[ \widehat{W} = \left\{ \emb(x,y): (x,y) \in E \right\} \subseteq X^k \times Y^k \]
By equation \eqref{winning_edges}, we observe that
 \[ |\widehat{W} \backslash W| \leq \eps |\widehat{W}| = \eps |E| \]
and hence, $W$ and $\widehat{W}$ are $\eps$-close. Note that some of the edges in $\widehat{W}$ may not exist in $E_H$. However, it will be useful to think of $\widehat{W}$ as a set of edges that induces a graph on repeated game vertices $\mathcal{H}_{\widehat{W}} = (im(f_X), im(f_Y), \widehat{W})$. While we use the notation $\mathcal{H}_{\widehat{W}}$ to indicate that it is a graph on repeated game vertices, it is again important to note that this graph is not a subgraph of $\mathcal{H}$. 

For the remainder of the proof, we will assume that $W$ and $\widehat{W}$ have no parallel edges, and that $\mathcal{H}_{\widehat{W}}$ is isomorphic to $G$. We note that this is not strictly true if $G$ has parallel edges - however, since we know that $G$ has at most $\eps |E|$ parallel edges so $\mathcal{H}_{\widehat{W}}$ is $\eps$-close to a graph that is isomorphic to $G$ even after taking out parallel edges. Hence, the same argument goes through by simply making $\eps$ slightly smaller. For a more detailed discussion of how we handle a small number of parallel edges, we refer the reader to the Remark at the end of this section.

We argue that very few vertices in $x \in X$ and $y \in Y$ are heavily represented in round $s$ of the sets of repeated game vertices $im(f_X)$ and $im(f_Y)$. Informally, we will use the uniform marginals property, which lets us conclude that there are many distinct base game edges $(x,y)$ in round $s$ of the edge set $W$, to conclude that there must be many distinct base game vertices in round $s$ of endpoints of $W$. To formalize the notion of a vertex $x$ being heavily represented in round $s$ of $im(f_X)$, define the weight of $x \in X$ be $w_x = |\{\oline{x} \in im(f_X): \oline{x}_s = x \}|$. Similarly, let the weight of $y \in Y$ be $w_y = |\{\oline{y} \in f_Y: \oline{y}_s = y \}|$. We argue that under the uniform marginals property of Theorem \ref{main_thm}, there cannot be many vertices with weight more than $2z = 2\frac{|H|}{|G|}$. 


\begin{proposition}
\label{marginal_count_edges}
Take $G$ and $H$ to be the base game and repeated game from Theorem \ref{main_thm}, and let $E$ denote the edge set of the base game $G$. For any fixed edge $(x,y) \in E$ and round $j \in [k]$, we have that  
\[
 \left| \left\{ (\oline{x}, \oline{y}) \in E_H : (\oline{x}_j, \oline{y}_j) = (x,y) \right\} \right| = z
\]
\end{proposition}
\begin{proof}
By applying the uniform marginals property of Theorem \ref{main_thm}, we observe that 
\begin{align*}
 \left| \left\{ (\oline{x}, \oline{y}) \in E_H : (\oline{x}_j, \oline{y}_j) = (x,y) \right\} \right| &= \Pr_{(\oline{x},\oline{y}) \in E_H} \left[ (\oline{x}_j, \oline{y}_j) = (x,y) \right] \cdot |E_H| \\
 &= \frac{|E_H|}{|E|}  = z
\end{align*}
\end{proof}

Below, say that a repeated game vertex $v \in im(f_X) \cup im(f_Y)$ is BAD if $v_s$ has weight more than $2z$.

\begin{lemma}
\label{few_bad_edges}
There are at most $2 \eps |E|$ repeated game edges in $W$ incident to BAD vertices.
\end{lemma}
For a fixed $x \in X$, define the set of vertices $PREIMG_x$ to be the set of edges in the repeated game that have $x$ in the $s^{th}$ coordinate, $PREIMG_x := \{ \oline{x} \in im(f_X): \oline{x}_s = x\}$. Notice that BAD vertices are exactly repeated game vertices that belong to $PREIMG_{x}$ for some $x \in X$ such that $w_x > 2z$. In what follows, we will let $d_X = d_Y = d$.

\begin{proof}[Proof (of Lemma \ref{few_bad_edges})]
For each $x' \in X$ with weight $w_{x'} > 2z$, we will argue that there are many edges in $\widehat{W} \backslash W$ incident to $PREIMG_{x'}$. Then, by noticing that $\widehat{W}$ and $W$ are $\eps$-close, we will be able to upper bound the number of vertices in $X$ with weight more than $2z$. Applying the uniform marginals property of $H$ from Theorem \ref{main_thm}, we will get an upper bound on the number of edges in $W$ incident to BAD vertices. 

Fix $x' \in X$ such that $w_{x'} > 2z$. We argue that there are at least $2z \cdot d_X$ edges in $\widehat{W}$ that are incident to $PREIMG_{x'}$. Recall that $\mathcal{H}_{\widehat{W}}$ is isomorphic to $\mathcal{G}$, the underlying graph of $G$. Specifically, this means that the degree of every member of $PREIMG_{x'}$ in $\mathcal{H}_{\widehat{W}}$ is exactly $d_X$, so in total there are $w_{x'} \cdot d_X \geq 2z \cdot d_X$ edges in $\widehat{W}$ incident to elements of $PREIMG_{x'}$. 

Recall from Proposition \ref{marginal_count_edges} that the uniform marginals property of the scheme $\scheme$ tells us that, for any fixed edge $(x,y) \in E$, the number of edges $(\oline{x}, \oline{y}) \in E_H$ such that $(\oline{x}_s, \oline{y}_s) = (x,y)$ is exactly $z$. Since $W \subseteq E_H$, we conclude that for any fixed $(x,y) \in E$ we have that
\[|\{ (\oline{x}, \oline{y}) \in W : (\oline{x}_s, \oline{y}_s) = (x,y)\}| \leq z \]
By fixing $x$ and summing over all $y$ such that $(x,y) \in E$, of which there are exactly $d_X$, we can see that 
\[|\{ (\oline{x}, \oline{y}) \in W : \oline{x}_s = x \}| \leq z \cdot d_X \]
for any fixed $x \in X$.
In other words, there can be at most $z \cdot d_X$ edges in $W$ incident to vertices in $PREIMG_{x}$, for any $x \in X$.

 Combining our lower bound of $2z \cdot d_X$ for the number of edges in $\widehat{W}$ incident to $PREIMG_{x'}$ and our upper bound of $z \cdot d_X$ for the number of edges in $W$ incident to $PREIMG_{x'}$, we see there are at least $2z \cdot d_X - z \cdot d_X = z \cdot d_X$ edges in $\widehat{W} \backslash W$ that touch $PREIMG_{x'}$, and that this is true for all $x'$ such that $w_{x'} > 2z$.

Noticing that there are not many edges in $\widehat{W} \backslash W$, we can upper bound the number of $X$ vertices with weight more than $2z$. For each vertex $x \in X$, let the variable $i_x$ denote the number of edges incident to the vertex set $PREIMG_x$ that are in $\widehat{W} \backslash W$. Since $|\widehat{W} \backslash W| \leq \eps |E|$, we get:
\begin{align*}
\eps |E| &\geq \sum\limits_{x \in X: w_x > 2z} i_x \\
&\geq |\{ x \in X: w_x > 2z\}| \cdot z \cdot d_X
\end{align*}
So we get that the number of base game vertices $x \in X$ with weight more than $2z$ is at most $\frac{\eps |E|}{z d_X}$. Reapplying the observation that there can be at most $z \cdot d_X$ edges incident to $PREIMG_x$ for any base game vertex $x$, we see that there can be at most $\eps |E|$ edges in $W$ incident to BAD vertices that live in $im(f_X)$. 

Repeating the proof for vertices in $Y$ shows there are at most $\frac{\eps |E|}{z d_Y}$ vertices in $Y$ with weight more than $2z$, and at most $\eps |E|$ edges in $W$ incident to BAD vertices that live in $im(f_Y)$. Union bounding over vertices in $im(f_X)$ and $im(f_Y)$ yields the result.
\end{proof} 
Lemma \ref{few_bad_edges} lets us remove all the bad vertices from $\mathcal{H}_{W}$, along with all the edges incident to them, and still have a graph with at least $(1 - 3 \eps)|E|$ edges. Call the resulting graph $\mathcal{H}'_{W} = ((X'_W, Y'_W), W')$. We remove the same vertices and the incident edges from $\mathcal{H}_{\widehat{W}}$ to get the graph $\mathcal{H}_{\widehat{W}'} = ((X'_W, Y'_W), \widehat{W}')$. Note that we still have $W' \subseteq \widehat{W}'$ and $|\widehat{W}' \backslash W'| \leq \eps |E|$, and since we did not remove many edges thanks to Lemma \ref{few_bad_edges}, we know that $|\widehat{W}'| \geq |W'| \geq (1 - 3\eps)|E|$. 

We would like to find a subset of vertices $S \subseteq X'_W$ such that every element of $\{x \in X: \exists \oline{x} \in S \text{ s.t. } \oline{x}_s = x\}$ have similar weights, and find an analogous subset $T \subseteq Y'_W$.
\begin{lemma}
\label{same_weights}
There are subsets $S \subseteq X'_W$ and $T \subseteq Y'_W$ such that:
\begin{enumerate}
\item $|\widehat{W}' \cap (S \times T)| \geq \frac{(1 - 6 \eps)|E|}{4\log^2(z)}$
\item $W' \cap (S \times T)$ is $2 \eps$-close to $\widehat{W}' \cap (S \times T)$.
\item There are integers $w_x^*, w_y^* \in \mathbb{Z}^+$ such that for any $x \in X \text{ s.t. } \oline{x}_s = x$ for some $\oline{x} \in S$ and $y \in Y \text{ s.t. } \oline{y}_s = y$ for some $\oline{y} \in T$, we have that $w_{x}^* \leq w_x \leq 2 w_x^*$ and $w_y^* \leq w_y \leq 2 w_y^*$.
\end{enumerate}
\end{lemma}
\begin{proof}
For each pair of positive integers $(i,j)$ such that $0 \leq i,j \leq \lceil \log(2z) \rceil -1$, let $S_i = \{ \oline{x} : \oline{x}_s = x \text{ for } x \in X \text{ s.t. } 2^i \leq w_x \leq 2^{i+1}\}$ and $T_j = \{ \oline{y}:\oline{y}_s = y, y \in Y, 2^j \leq w_y \leq 2^{j+1}\}$. Note that the sets $\{ \widehat{W}' \cap (S_i \times T_j): {1 \leq i, j \leq \lceil \log(2z) \rceil -1}\}$ form a partition of the edges in $\widehat{W}'$, since we removed all BAD vertices and incident edges earlier. We will call a pair $(i,j)$ bad if $W' \cap (S_i \times T_j)$ is more than $2 \eps$-far from the edge set $\widehat{W}' \cap (S_i \times T_j)$, and good otherwise. 

Since $\widehat{W}' \backslash W'$ has size at most $\eps |E|$, we can upper bound the size of the set 
\[\bigcup_{i,j: (i,j) \text{ is bad}}  \widehat{W}' \cap (S_i \times T_j) \] 
as follows:
\begin{align*}
\eps |E| &\geq |\widehat{W}' \backslash W'| \\
&= \sum\limits_{0 \leq i,j \leq \lceil \log(2z) \rceil -1} \left| \left(\widehat{W}' \cap (S_i \times T_j) \right) \backslash \left(W' \cap (S_i \times T_j) \right) \right| \\
&\geq \sum\limits_{(i,j): (i,j) \text{ is bad}} \left|\left(\widehat{W}' \cap (S_i \times T_j) \right) \backslash \left(W' \cap (S_i \times T_j) \right) \right| \\
&\geq 2 \eps \sum\limits_{(i,j): (i,j) \text{ is bad}} \left|\widehat{W}' \cap (S_i \times T_j) \right| \\
&= 2 \eps \left| \bigcup_{i,j: (i,j) \text{ is bad}}  \widehat{W}' \cap (S_i \times T_j) \right|
\end{align*}
Therefore, we can conclude that 
\[\bigcup_{i,j: (i,j) \text{ is bad}}  \widehat{W}' \cap (S_i \times T_j) \] 
has at most $|E|/2$ edges, and therefore
\[\bigcup_{i,j: (i,j) \text{ is good}}  \widehat{W}' \cap (S_i \times T_j) \] 
has at least $(\frac{1}{2} - 3 \eps)|E|$ edges. Since $i$ and $j$ range from 0 to $\lceil \log(2z) \rceil -1$, there are at most $(\log(2z)+1)^2 \leq 2\log^2(z)$ good pairs, so there is some choice of positive integers $i^*$ and $j^*$ such that $\widehat{W}' \cap (S_{i^*} \times T_{j^*})$ has at least $\frac{(1 - 6 \eps)|E|}{4\log^2(z)}$ edges and $(i^*, j^*)$ is good, so $W' \cap (S_{i^*} \times T_{j^*})$ is $2 \eps$-close to $\widehat{W}' \cap (S_{i^*} \times T_{j^*})$. Taking $S := S_{i^*}$, $T := T_{j^*}$, $w_x^* = 2^{i^*}$, and $w_y^* = 2^{j^*}$ completes the proof.
\end{proof}

Note that Property 2 of Lemma \ref{same_weights} allows us to lower bound the number of edges in $W' \cap (S \times T)$. Since $W' \cap (S \times T)$ is $2 \eps$-close to $\widehat{W}' \cap (S \times T)$, we get that 
\[ |W' \cap (S \times T)| \geq (1 - 2\eps) |\widehat{W}' \cap (S \times T)| \geq \frac{(1 - 2\eps)(1 - 6 \eps) |E|}{4\log^2(z)}\]
Furthermore, note that each vertex in $X'_W$ has degree at most $d_X$ in $\mathcal{H}_{\widehat{W}'}$, and furthermore each vertex in $Y'_W$ has degree at most $d_Y$ in $\mathcal{H}_{\widehat{W'}}$. This can be seen by noting that $\mathcal{H}_{\widehat{W}}$ is isomorphic to $\mathcal{G}$, and we removed some edges when we removed BAD vertices. Combining this with the fact that $|E| = |X|d_X = |Y|d_Y$ and applying the lower bound on $|\widehat{W}' \cap (S \times T)|$, we can lower bound the sizes of the vertex sets $S$ and $T$ from Lemma \ref{same_weights}. Specifically, we get that $|S| \geq \frac{(1 - 6 \eps)|X|}{4\log^2(z)}$ and  $|T| \geq \frac{(1 - 6 \eps)|Y|}{4\log^2(z)}$.  \\


Now, in accordance with the proof outline, we would like to retrieve a large subset $M \subseteq X_W'$ such that, for all $\oline{x}, \oline{x}' \in M$ such that $\oline{x} \neq \oline{x}'$, we have that $\oline{x}_s \neq \oline{x}'_s$. Similarly, we want a large subset $N \subseteq Y_W'$ such that, for all $\oline{y}, \oline{y}' \in N$ such that $\oline{y} \neq \oline{y}'$, we have that $\oline{y}_s \neq \oline{y}'_s$.


\begin{lemma}
\label{one_to_one_rect}
There are sets $M \subseteq X'_W$ and $N \subseteq Y'_W$ such that:
\begin{enumerate}
\item $M$ contains at most one element of the set $\{ \oline{x} \in X'_W: \oline{x}_s = x\}$ for any fixed $x \in X$. Also, $N$ contains at most one element of the set $\{ \oline{y} \in Y'_W: \oline{y}_s = y \}$ for any fixed $y \in Y$.
\item $W' \cap (M \times N)$ is $8 \eps$-close to $\widehat{W}' \cap (M \times N)$
\item $|M| \geq \frac{(1 - 6 \eps)|X|}{8z\log^2(z)}$ and $|N| \geq \frac{(1 - 6 \eps)|Y|}{8z\log^2(z)}$
\end{enumerate}
\end{lemma}
\begin{proof}[Proof (of Lemma \ref{one_to_one_rect}):]

Start with the sets $S \subseteq X'_W$ and $T \subseteq Y'_W$ as well as $w_x^*$ and $w_y^*$ from Lemma \ref{same_weights}. We will show the existence of $M \subseteq S$ and $N \subseteq T$ with the desired properties. Set $w_x^{max} = \sf{min}(2w_x^*, 2z)$, and similarly set $w_y^{max} = \sf{min}(2w_y^*, 2z)$. We note that, for any $x \in X$ such that $\{\oline{x} \in S: \oline{x}_s = x \}$ is nonempty, we know that $w_x \leq w_x^{max}$, from Lemma \ref{same_weights} and our removal of BAD vertices, and we have a similar condition for vertices $y \in Y$.

For each vertex $x \in X$ such that $\{\oline{x} \in S: \oline{x}_s = x \}$ is nonempty, label each vertex in the set $\{\oline{x} \in S: \oline{x}_s = x \}$ as $x^1, \ldots, x^{w_x^{max}}$. Since $|\{\oline{x} \in S: \oline{x}_s = x \}| = w_x$ and $w_x^{max} \geq w_x$, each vertex gets a label. However, note that it is possible that $w_x^{max} > w_x$, in which case we wrap around with our labeling. Since Lemma \ref{same_weights} gives us that $w_x \geq w_x^* \geq w_x^{max} / 2$, we know that any vertex receives at most two labels. Similarly, for each vertex $y \in Y$ such that $\{\oline{y} \in T: \oline{y}_s = y \}$ is nonempty, label each vertex in the set $\{\oline{y} \in T: \oline{y}_s = y \}$ as $y^1, \ldots, y^{w_y^{max}}$. Once again we observe that every vertex gets a label and any vertex receives at most two labels.

For $i \in \mathbb{Z}$ such that $1 \leq i \leq w_x^{max}$, let
 \[M_i = \bigcup_{x \in X: \{\oline{x} \in S: \oline{x}_s = x \} \neq \emptyset } x^i\]
 Similarly, for $j \in \mathbb{Z}$ such that $1 \leq j \leq w_y^{max}$, let 
 \[N_j = \bigcup_{y \in Y: \{\oline{y} \in T: \oline{y}_s = y \} \neq \emptyset} y^j \]
 Since any vertex $\oline{x} \in S$ received at most two labels, note that it is present in $M_i$ for at most two choices of $i$. Similarly, any vertex $\oline{y} \in T$ is present in $N_j$ for at most two choices of $j$. Consider the sets of pairs of vertices given by
\begin{equation*} 
\label{final_partition}
\left\{ M_i \times N_j : 1 \leq i \leq w_x^{max}, 1 \leq j \leq w_y^{max} \right\} 
\end{equation*}

The union of these sets contains $S \times T$. Hence, every edge in $W' \cap (S \times T)$ is in $W' \cap (M_i \times N_j)$ for some choice of $i$ and $j$. Similarly, every edge in $ \widehat{W}' \cap (S \times T)$ is in $\widehat{W}' \cap (M_i \times N_j)$ for some choice of $i$ and $j$. Furthermore, as we noticed earlier, any vertex $\oline{x} \in S$ is in $M_i$ for at most two choices of $i$ and any vertex $\oline{y} \in T$ is in $N_j$ for at most two choices of $j$. Therefore, any fixed pair of repeated game vertices $(\oline{x}, \oline{y})$ only appears in $M_i \times N_j$ for at most 4 choices of $(i,j)$. Hence we know that 
\[ \sum_{i,j} |\widehat{W}' \cap (M_i \times N_j)| \geq \left|\bigcup_{i,j}\widehat{W}' \cap (M_i \times N_j) \right| \geq |\widehat{W}' \cap (S \times T)| \]
and that
\[ 
\sum_{i,j} |(\widehat{W}' \backslash W') \cap (M_i \times N_j)| \leq 4 |(\widehat{W}' \backslash W') \cap (S \times T)| \leq 8 \eps |\widehat{W}' \cap (S \times T)|
\]

Therefore, the average fraction of edges in $\widehat{W}' \cap (M_i \times N_j)$ that are also in $(\widehat{W}' \backslash W') \cap (M_i \times N_j)$ is at most $8\eps$. Therefore, there must be a fixing of $i^*$ and $j^*$ such that the set of edges $W' \cap (M_{i^*} \times N_{j^*})$ is $8 \eps$-close to $\widehat{W}' \cap (M_{i^*} \times N_{j^*})$, by pigeonhole. Furthermore, since $w_x^{max}, w_y^{max} \leq 2z$, we know that 
\[ |M_{i^*}| = \frac{|S|}{w_x^{max}} \geq \frac{|S|}{2z} \geq \frac{(1 - 6 \eps)|X|}{8z\log^2(z)}\]
and
\[ |N_{j^*}| = \frac{|T|}{w_y^{max}} \geq \frac{|T|}{2z} \geq \frac{(1 - 6 \eps)|Y|}{8z\log^2(z)}\]

By letting $M := M_{i^*}$ and $N := N_{j^*}$, we conclude the proof.
\end{proof}

We notice that Lemma \ref{one_to_one_rect} also gives us an explicit lower bound on the number of edges in $W' \cap (M \times N)$. Since none of the vertices in $M$ or $N$ are BAD by construction, we know that
\begin{equation}
\label{didnt_remove_good_vertices}
\widehat{W}' \cap (M \times N) = \widehat{W} \cap (M \times N)
\end{equation}
since to get from $\widehat{W}$ to $\widehat{W}'$ we only removed edges incident to BAD vertices. Also recall that $\mathcal{H}_{\widehat{W}}$ is isomorphic to $\mathcal{G}$, and therefore has the same expansion property as $G$, given by the expansion property of Lemma \ref{random_csp}. Since Property 3 of Lemma \ref{one_to_one_rect} lower bounds the size of $M$ and $N$, we can apply the expansion property of $\mathcal{H}_{\widehat{W}}$ to get:
\begin{equation}
\label{apply_expansion}
\left| \widehat{W} \cap (M \times N) \right| \geq (1 - \eps) \frac{d_X |M| |N|}{|Y|}
\end{equation}

By combining Item 2 of Lemma \ref{one_to_one_rect} and Equations \ref{didnt_remove_good_vertices} and \ref{apply_expansion}, we see that
\begin{equation}
\label{lower_bound_on_num_edges}
\left| W' \cap (M \times N) \right| \geq (1 - 8 \eps) \left| \widehat{W}' \cap (M \times N) \right| \geq (1 - 8 \eps)(1 - \eps) \frac{d_X |M| |N|}{|Y|}
\end{equation}

Now we can prove the main theorem.
\begin{proof} [Proof of Theorem \ref{main_thm}]

Take $M \subseteq X'_W$ and $N \subseteq Y'_W$ to be the sets given by Lemma \ref{one_to_one_rect}. Let $M_s = \{ \oline{x}_s: \oline{x} \in M \}$ and $N_s = \{ \oline{y}_s: \oline{y} \in N\}$ be the sets that result from projecting the repeated game vertices in $M$ and $N$ onto round $s$. Due to Property 1 of Lemma \ref{one_to_one_rect}, for every pair of vertices $\oline{x}^1, \oline{x}^2 \in M$, we know that $\oline{x}^1_s \neq \oline{x}^2_s$. Similarly, for every pair of vertices $\oline{y}^1, \oline{y}^2 \in N$, we know that $\oline{y}^1_s \neq \oline{y}^2_s$. Therefore, we see that 
\[ |M_s| = |M| \geq \frac{(1 - 6 \eps)|X|}{8z\log^2(z)} \]
and 
\[ |N_s| = |N| \geq  \frac{(1 - 6 \eps)|Y|}{8z\log^2(z)}
\] 
where the lower bounds follow from Property (3) of Lemma \ref{one_to_one_rect}. Furthermore, any assignment to vertices in $M$ and $N$ corresponds uniquely to an assignment to $M_s \subseteq X$ and $N_s \subseteq Y$, by simply restricting the assignment to vertices in $M$ and $N$ to round $s$. 

 Since $W' \cap (M \times N) \subseteq W_{\{s\}}$, we know that every edge in $W' \cap (M \times N)$ is satisfied in round $s$ by the assignment $\psi^H$. By restricting $\psi^H_X$ to $M$ and $\psi^H_Y$ to $N$, considering only round $s$ of this assignment, and applying the fact that each edge in $W' \cap (M \times N)$ corresponds to a unique edge in round $s$, we retrieve an assignment that satisfies $|W' \cap (M \times N)|$ edges in the rectangular subgame $G_{M_s \times N_s}$. By applying the expansion property of $\mathcal{G}$, we can upper bound the number of edges in this rectangle:
 \[ 
 \left| E \cap (M_s \times N_s) \right| \leq (1 + \eps) \frac{d_X |M| |N|}{|Y|}
 \]
 Hence, by applying Equation \ref{lower_bound_on_num_edges}, the fraction of constraints in $G_{M_s \times N_s}$ satisfied by our assignment is at least
 \[ 
 \frac{|W' \cap (M \times N)|}{\left| E \cap (M_s \times N_s) \right|} \geq \frac{(1 - 8 \eps)(1 - \eps)}{1 + \eps} > 1 - 11\eps \geq \text{val}(G) + \eps
 \]
 due to our assumption on $\val(G)$. This, along with the fact that
 \begin{align*}
  \delta &= \frac{1}{16\Phi_{\scheme, \Sigma}(n)\log^2(\Phi_{\scheme, \Sigma}(n))} \\
  &\leq \frac{1}{16z\log^2(z)} \\
  &\leq \frac{1-6\eps}{8z\log^2(z)}
  \end{align*}
 means that we contradict the fact that $G$ is $(\delta, \eps)$-fortified.
\end{proof}

\paragraph{Remark about handling Parallel Edges.}

We end this section by remarking on why parallel edges can be problematic and how we handle them. In the last step of the proof, we lift round $s$ of the assignment $\psi^H$ on the rectangle $M \times N \subseteq X^k \times Y^k$ to an assignment for the rectangular subgame $G_{M_s \times N_s}$. We argued that each edge in the edge set $W' \cap (M \times N)$ lifted to a distinct edge in $G_{M_s \times N_s}$, by virtue of the fact that each vertex in $M$ and $N$ is distinct in round $s$. This is valid when $W'$ is a set of edges, rather than a multiset; however, if we considered $W'$ to be a multiset and it had parallel edges, this may no longer be true. Two distinct, but parallel, edges in $W'$, could lift to only one distinct edge in $G_{M_s \times N_s}$, in which case we lose an edge! In the case when the number of parallel edges is small (i.e. $\leq \eps |E|$), we can prevent this inconvenience by effectively ignoring the parallel edges. 

Concretely, we can make $W$ a multiset that has no parallel edges by ignoring parallel edges in the domain of the embedding map (i.e. each pair of vertices that appears in $W$ has multiplicity 1). Since the number of parallel edges is small, we will still have $|W| \geq (1 - 2 \eps) |E|$ and that $W$ is $2 \eps$-close to a multiset of edges $\widehat{W}$, where $\mathcal{H}_{\widehat{W}}$ is isomorphic to $G$, parallel edges and all. By naturally extending the notion of $\eps$-closeness to multisets, and defining the intersection of a multiset and a set to preserve multiplicity (i.e. $\{1, 1, 1, 2\} \cap \{1\} = \{1, 1, 1\}$), our arguments naturally extend to this case without any further change. 

 For completeness, we conclude with a note about the number of parallel edges in the random games we provide in Appendix~\ref{sec:random_games_fortified}. As long as $200d^2 < \eps |E|$, the random games we generate have sufficiently few parallel edges for our Main Theorem to apply. When $200 d^2 > \eps |E|$, since $|E| = d |X|$, we must have that $d = \Omega(|X|)$. For this regime of $d$, we can simply use a free game with random constraints. It can be seen by the analysis in Claim~\ref{random_makes_fortified} of Appendix~\ref{sec:random_games_fortified} that this game is sufficiently fortified and satisfies the conditions we need for the Main Theorem.

%

\section{Robust embeddings in existing proofs of parallel repetition}
\label{sec:raz-hol}
Here we show that Raz's proof of the parallel repetition theorem directly implies a robust embedding from $G$ into $G^k$. Raz's proof was significantly simplified by Holenstein in ~\cite{Hol07}. Throughout this section, we will follow Rao's presentation ~\cite{Rao09} of Raz's proof with Holenstein's simplification. From now on, we will refer to this proof as the Raz-Holenstein proof of parallel repetition.

The engine behind the Raz-Holenstein proof of parallel repetition theorem is the following lemma. 

	\begin{lemma}[Main lemma of \cite{Rao09}]
	\label{rao_main_lemma}
	Let $C \subseteq [k]$. Let $G$ be a game with $\val(G) = 1 - \eps$, where one of the provers gives answers from a set of size $2^c$, and there exists a strategy $\psi$ for $G^k$ under which
	$$
		\Pr [ W_C] \geq 2^{- \frac{\eps^2(k - |C|)}{34^2} + |C|c}.
	$$
	Then there exists an $i \notin C$ such that $\Pr [ W_i | W_C]  \leq \val(G) + \eps/2 = 1 - \eps/2$.
	\end{lemma}
	Here, we use $W_C$ to denote the event that the provers succeed in the rounds indexed by $C$; note that this event depends on the strategy used by the provers. We use $W_{\{i \}}$ to denote the event that the provers win round $i$.
	
	From Lemma~\ref{rao_main_lemma}, the parallel repetition theorem follows in a straightforward manner. We want to show that the probability of winning every round in $G^k$, $\Pr [W_{[k]}]$, is $2^{-\gamma k}$ for some constant $\gamma$. We accomplish this by iteratively building a subset of rounds $C \subseteq [k]$ such that either $\Pr [W_C] < 2^{-\gamma k}$ (in which case we're done, because $\Pr[W_{[k]}] \leq \Pr[W_C]$), or otherwise, by upper bounding $\Pr [W_{\{ i\}}| W_C]$ for some $i \notin C$, we conclude that $\Pr [W_{\{ i \} \cup C}] < (1 - \eps/2) \Pr [W_C]$ and recurse with $C' = C \cup \{ i \}$.  After repeatedly applying this lemma at most $\beta k$ times, we can conclude that $\Pr [W_{[k]}] \leq \max \{ 2^{- \gamma k}, (1 - \eps/2)^{\beta k} \}$, which proves the parallel repetition theorem.

Implicit in the proof of Lemma~\ref{rao_main_lemma} is the following lemma, which demonstrates the existence of a robust embedding of $G$ into $G^k$.

\begin{lemma}[Implicit Lemma in \cite{Rao09}]
\label{implicit_lemma}
	Let $C \subseteq [k]$ be such that
	$$
		\Pr [W_C] \geq 2^{- \frac{\eps^2(k - |C|)}{34^2} + |C| \cdot c}.
	$$
	Then there exist randomized maps $g_X: R \times X \to X^k$ and $g_Y : R \times Y \to Y^k$ for some finite set $R$ such that
	\begin{enumerate}
	\item For all $r \in R$, there exists a round $i \in [k] \backslash C$ such that for all $(x,y) \in X \times Y$, we have $g_X(r,x)_i = x$ and $g_Y(r,y)_ i = y$.
	\item The distribution of $(g_X(r,x),g_Y(r,y))$ over a uniformly chosen $r \in R$ and $(x,y) \in E$ is $\eps/2$-close in statistical distance to the distribution of $(\oline{x},\oline{y})$ in $G^k$ when conditioned on the event $W_C$.
	\end{enumerate} 
		
\end{lemma}
%
%

First, we claim that Implicit Lemma very directly implies the existence of a robust embedding from $G$ into a coordinate of $G^k$. Indeed, assume that Lemma~\ref{implicit_lemma} is true. Let $\gamma: [k] \to \mathbb{R}$ be defined as $\gamma(t) = 2^{- \frac{\eps^2(k - t)}{34^2} +t \cdot c}$. For each $C$, if $\Pr[W_C] < \gamma(|C|)$, then we let $\emb$ be an arbitrary embedding map. If $\Pr[W_C] \geq \gamma(|C|)$, then there exist randomized maps $g_X$ and $g_Y$ satisfying Property 2 of Lemma~\ref{implicit_lemma}. Furthermore, by averaging, there must exist an $r^* \in R$ such that $\Pr_{(x,y) \in E} [ (g_X(r^*,x), g_Y(r^*,y)) \in W_C ] \geq 1 - \eps/2$. Let $\emb(x,y) = (g_X(r^*,x),g_Y(r^*,y))$. This shows that there is a $(\gamma(t),\eps/2)$-robust embedding of $G$ into a coordinate of $G^k$.

Furthermore, the Implicit Lemma also implies Lemma~\ref{rao_main_lemma}:

\begin{proof}[Proof that Lemma \ref{implicit_lemma} Implies Lemma \ref{rao_main_lemma}]
We assume the Implicit Lemma. Let $C \subseteq [k]$ be as described in the statement of the lemma, and let $g_X: R \times X \to X^k$ and $g_Y : R \times Y \to Y^k$ be the randomized embedding maps.

We now describe a strategy for the provers to play the base game $G$. The provers are given $x$ and $y$ where $(x,y)$ is a uniform edge from $E$. The two provers, using shared randomness, sample a uniformly random $r \in R$. The first prover computes $\oline{x} = g_X(r,x)$ and then $\oline{a} = \psi_X(\oline{x})$. The first prover answers with $\oline{a}_i$.  The second prover computes $\oline{y} = g_Y(r,y)$ and then $\oline{b} = \psi_Y(\oline{y})$. The second prover answers with $\oline{b}_i$. 

Since this is a strategy for $G$, the probability that the provers win is at most $\val(G)$. On the other hand, since the distribution of $(\oline{x},\oline{y})$ generated by the provers is $\eps/2$-close to the distribution of $(\oline{x}_i,\oline{y}_i)$ in the subgame $W_C$, we have that the probability the provers win using this strategy is at least $\Pr[ W_i | W_C ] - \eps/2$.

Thus we have $\Pr[W_i | W_C] \geq \val(G) + \eps/2$.

%
%
%
\end{proof}

Finally, for completeness, we give a high-level sketch of how the Implicit Lemma is proved. This argument follows the Raz-Holenstein proof of the parallel repetition theorem. Let $C \subseteq [k]$ be a set of coordinates such that $\Pr[W_C] \geq \gamma(|C|)$. Let $\oline{X},\oline{Y}, \oline{A},\oline{B}$ denote the random variables corresponding to the questions and answers of the provers when playing $G^k$. The randomized maps $g_X$ and $g_Y$ will correspond to a protocol where the first prover (who receives a question $x\in X$) and the second prover (who receives $y\in Y$) utilize shared randomness $R$ in order to agree on a coordinate $i \in [k] \backslash C$, and produce questions $\oline{x} \in X^k$ and $\oline{y} \in Y^k$, respectively, so that $\oline{x}_i = x$, $\oline{y}_i = y$, and furthermore, their outputs $(\oline{x},\oline{y})$ are (approximately) distributed the same way as $(\oline{X},\oline{Y})$ are, conditioned on the event $W_C$.

The key to this protocol, and the cornerstone of the Raz-Holenstein parallel repetition theorem is the \emph{dependency-breaking random variable} $Q$, which resides in the same probability space as $\oline{X},\oline{Y}, \oline{A},\oline{B}$. This random variable has the property that, conditioned on $Q$ and (say) the first prover's question $x$, the repeated questions $\oline{X}$ and $\oline{Y}$ are \emph{independent}. Furthermore, the variable $Q$ has the remarkable property that the following distributions are close in statistical distance\footnote{Technically speaking, they are close on average over $i$, $x$, and $y$.}:
$$
p(Q | \oline{X}_i = x, W_C) \approx p(Q | \oline{X}_i = x, \oline{Y}_i = y, W_C) \approx p(Q | \oline{Y}_i = y, W_C)
$$
where by $p(Q | \oline{X}_i = x, W_C)$, for example, we mean the distribution of $Q$ conditioned on $\oline{X}_i = x$ and the event $W_C$. Using a beautiful technique called \emph{correlated sampling}, the two provers can use shared randomness to (approximately) jointly sample $Q$ from the distribution $p(Q | \oline{X}_i = x, \oline{Y}_i = y, W_C)$, even though they only know one of $x$ or $y$, but not both. 

Since $i$ was picked randomly, with high probability the distribution of $(\oline{X}_i,\oline{Y}_i)$ conditioned on $W_C$ will also be close to the distribution of questions in the original game $G$. This implies that the final distribution of the output of the maps $g_X$ and $g_Y$ will be close to the distribution of $(\oline{X},\oline{Y})$ conditioned on $W_C$, which is what we desired.

In addition to the Raz-Holenstein proof, nearly all subsequent proofs of parallel repetition fall into the embedding framework, including the works of Rao~\cite{Rao09}, Moshkovitz~\cite{M14}, and Braverman-Garg~\cite{BG}. We also believe that the analytical proof of parallel repetition given by Dinur and Steurer in ~\cite{DinurSteurer} falls under this framework. 


\section{A Contrived Example for Derandomized Parallel Repetition}
\label{contrived_example}
In this section we show that we cannot hope to obtain a strong no-go theorem that rules out any derandomized parallel repetition in the high degree regime, the same spirit as the result of Feige and Kilian. This is because there \emph{is} a parallel repetition scheme that, when applied to some games, actually reduces the value in a very randomness-efficient manner. We use Dinur's graph powering gap amplification scheme, which is a highly randomness-efficient parallel repetition scheme. For any $\eps > 0$, we construct a game $G$ with value  $\geq 1/8$, such that the application of graph powering to $G$ yields a game $H$ with value at most $\eps$, and the randomness complexity of $H$ is $\log |G| + f(1/\eps)$ for some function $f$. If $|G|$ is a growing parameter, then for constant $\eps$, this is much less than $O(\log \frac{1}{\eps}) \cdot \log |G|$, the randomness complexity that would be needed if we used standard parallel repetition to reduce the value from $1/8$ to $\eps$.

Unfortunately, this doesn't show that graph powering is a useful derandomized parallel repetition scheme\footnote{In fact, Bogdanov constructs games for which graph powering fails to achieve any error reduction at all~\cite{Bogdanov}.}. The game $G$ is constructed by first taking a game $G_{low}$ with value $\eps$, and ``hiding'' it in a high value game $G$ with value at least $1/8$. The game $H$ produced by graph powering ``uncovers'' $G_{low}$, and thus $\val(H) \leq \val(G_{low}) \leq \eps$. However, intuitively the error reduction was not obtained by graph powering per se, but rather came from a ``planted'' game that had low value to begin with. This shows that the degree-dependent lower bound of Feige and Kilian is in a sense tight, and thus to obtain stronger no-go results for derandomized parallel repetition, we turn to investigating proof strategies, which is the focus of our paper.
\subsection{The Derandomized Parallel Repetition Scheme: Graph Powering}
\label{gap_amplification}
Specifically, the derandomized parallel repetition scheme we use is graph powering, well-known from the gap amplification scheme of Dinur. This transforms a graph $G = (V,E)$ to a graph $G^{*t} = (V', E')$. In this graph, we have that $V' = V$, and each vertex $v \in V'$ intuitively corresponds to the ``cloud'' of vertices reachable from $v \in V$ in $t$ steps. Furthermore, each edge in $E'$ corresponds to a $(2t+1)$-step random walk in $E$. The prover is supposed to give each vertex $v' \in V'$ a super-label that contains labels for each of the vertices in its ``cloud,'' and each edge $e' = (u', v') \in E'$ checks that: 1) the labels to $u'$ and $v'$ are valid and 2) there is consistency in the labels of all the vertices shared between the cloud of $u'$ and the cloud of $v'$.

The graph powering method described above is a form of derandomized parallel repetition. If we let $d$ denote the maximum degree of the graph $G$, selecting a random edge in $G^{*t}$ takes $\log |V| + (2t + 1) \log d$ bits of randomness, as edges in $G^{*t}$ are simply $(2t+1)$-length random walks. Note that with $t$ and $d$ being constant, this is an extremely randomness efficient way to ask many questions. The main problem with using this as a derandomized parallel repetition scheme is that it is unclear how to prove that the value of $G^{*t}$ is decaying with increasing $t$. 
However, in this section we will create games $G$ for which the value of $G^{*t}$ is significantly lower than the value of $G$, and hence be able to use graph powering as derandomized parallel repetition. In fact, we will only need to focus on the case where $t=2$: that is, in this section, we will construct games $G$ where the value of $G^{*2}$ is much lower than the value of $G$. 

We also observe that the alphabet size of $G^{*t}$ is $|\Sigma|^{d^t}$ (since we are asking for labels to all the vertices reachable in $t$ steps from a vertex $v$). In our construction, $|\Sigma|$, $d$, and $t$ are all constant (relative to the size of the game), and thus the alphabet size is constant.


\subsection{A Sketch of the Construction}
The rough outline of the construction is as follows:
\begin{enumerate}
\item Start with a two prover game $G' = (X', Y', E', \Sigma, \mathcal{C})$, where $\mathcal{C}$ denotes the constraints on the game $G'$, that has low constant value $\varepsilon$, has a constant sized alphabet, and has constant degree. 
\item Use \emph{composition} to transform $G'$ into a game $G$ over the alphabet $\{0,1\}^3$. An exposition on composition can be found in Section 5 of~\cite{Dinur}. Roughly speaking, by composition we mean that we replace each constraint in the game $G'$ with a gadget that encodes the constraint, but is itself a game over alphabet $\{ 0,1\}^3$. Such gadgets are called \emph{assignment testers}, and have a size that depends only on the alphabet size of $G'$. The game $G$ that we get after composition necessarily has high value, as a random strategy satisfies at least $1/8$th of the constraints. More details can be found in Section \ref{contrived-step-2} below.
\item Use graph powering on $G$ to get the game $G^{*2}$, which will have value at most that of $G'$. 
\end{enumerate}
This construction works by using composition to hide the low value game $G'$ inside the high-value game $G$. However, the hiding was performed in a local fashion that can easily be uncovered by graph powering. Namely, the game $G^{*2}$ will contain constraints of $G'$, and hence have low value. Furthermore, due to the constant degree and alphabet size of $G'$, the game $G^{*2}$ will have very low randomness complexity -- no more than the randomness complexity of $G$ plus an additive constant. We now go into each step in further detail.

\subsection{Step 1: A Game with Low Value}
\label{contrived-step-1}
We start with a two player game $G' = (X', Y', E', \Sigma, \mathcal{C})$ with $\text{val}(G') < \varepsilon$, and the alphabet size and degree are functions of $1/\eps$. Since we think of $\eps$ as a constant, the alphabet size and degree are also constant.

\subsection{Step 2: Composition}
\label{contrived-step-2}
Recall that our goal in this section is to transform the game $G'$ into a game $G$ over the alphabet $\{ 0,1\}^3$. For this we will use composition with \emph{assignment testers} as described in Definition 5.1 in~\cite{Dinur}. We define assignment testers below:

\begin{definition}[Assignment Tester, Definition 2.2 from~\cite{Dinur}]
\label{def:assignment_tester}
An Assignment Tester with alphabet $\Sigma_0$ and rejection probability $\varepsilon > 0$ is an algorithm $\mathcal{P}$ whose input is a circuit $\Phi$ over Boolean variables $X$, and whose output is a constraint graph $G = ((V, E), \Sigma_0, \mathcal{C})$ such that $V \supset X$ and the following hold. Let $V' = V \backslash X$, and let $a: X \to \{0,1 \}$ be an assignment. 
\begin{itemize}
\item (Completeness) If $a \in \text{SAT}(\Phi)$, there exists $b: V' \to \Sigma_0$ such that $\text{UNSAT}_{a \cup b}(G) = 0$.
\item (Soundness) If $a \not\in \text{SAT}(\Phi)$, then for all $b: V' \to \Sigma_0$, we have $\text{UNSAT}_{a \cup b}(G) \geq \varepsilon \cdot \text{rdist}(a, \text{SAT}(\Phi))$.
\end{itemize}
where $\text{rdist}(a, S) = \text{min}_{s \in S} \frac{|a \oplus s|}{|V|}$ denotes the minimum relative Hamming distance between $a$ and elements of the set $S$, $\text{SAT}(\Phi)$ is the set of satisfying inputs to $\Phi$, and $\text{UNSAT}_{a \cup b}(G)$ is the fraction of constraints of $G$ that are unsatisfied by the assignment induced by $a$ and $b$.
\end{definition}
Additionally, Theorem 5.1 of~\cite{Dinur} gives us that there are explicit assignment testers over $\{0,1\}^3$ for a certain $\varepsilon > 0$. 

Using assignment testers, we can describe the composition of a game $G$ and an assignment tester $\mathcal{P}$. For this, we will use an error correcting code $e: \Sigma \to \{ 0,1\}^{\ell}$, where $\log_2 |\Sigma| \leq \ell \leq c \cdot \log_2 |\Sigma|$ for some constant $c$. 
\begin{definition}[Composition, Definition 5.1 from~\cite{Dinur}]
\label{def:composition}
Let $G = ((V, E), \Sigma, \mathcal{C})$ be a constraint graph and let $\mathcal{P}$ be an assignment tester. Let $e: \Sigma \to \{ 0,1\}^{\ell}$ be an encoding as described above with relative distance $\rho > 0$. The constraint graph $G \circ \mathcal{P} = ((V', E'), \Sigma_0, \mathcal{C}')$ is defined in two steps:
\begin{itemize}
\item (Robustization): First, we convert each constraint $c(e) \in \mathcal{C}$ to a circuit $\tilde{c}(e)$ as follows. For each variable $v \in V$, let $[v]$ be a fresh set of $\ell$ Boolean variables. For each edge $e = (v, w) \in E$, $\tilde{c}(e)$ will be a circuit on $2 \ell$ Boolean variables $[v] \cup [w]$ that outputs 1 iff the assignment for $[v] \cup [w]$ is a legal assignment for $v$ and $w$ that would have satisfied the constraint $c$ on $(v,w)$.
\item (Composition): Run the assignment tester $\mathcal{P}$ on each $\tilde{c}(e)$. Let $G_e = ((V_e, E_e), \Sigma_0, \mathcal{C}(e))$ denote the resulting constraint graph, and recall that $[v] \cup [w] \subset V_e$. Assume, wlog, that $E_e$ has the same cardinality for each $e$. Define the new constraint graph $G \circ \mathcal{P} = ((V', E'), \Sigma_0, \mathcal{C}')$ by
\[ V' = \bigcup_{e \in E} V_e \hspace*{3cm} E' = \bigcup_{e \in E}E_e \hspace*{3cm} \mathcal{C}' = \bigcup_{e \in E}\mathcal{C}_e\]
\end{itemize}
\end{definition}
As noted in~\cite{Dinur}, the output graph $G_e$ of an assignment tester $\mathcal{P}$ when it is used in composition above has size that depends only on the alphabet size of the game $G'$, which is a constant. Hence, the size of $G_e$ is also a constant. Furthermore, it can be seen from Definitions \ref{def:assignment_tester} and \ref{def:composition} that $G_{(u,v)}$ can have all its constraints satisfied if and only if the assignments given to $[u]$ and $[v]$ are legal assignments for $u$ and $v$ that satisfy the constraint $c((u,v))$.  

We will consider the modified assignment tester $\mathcal{P}'$, which acts as follows. It runs $\mathcal{P}$ on the input, and looks at the resulting constraint graph $H$. It then adds all missing edges to $H$ to create a complete graph $\overline{H}$, and puts trivially satisfied constraints on all of them. It can be seen that if $H$ had constant size, then so does $\overline{H}$. Note that the constraints of $H$ are all satisfiable if and only if the constraints of $\overline{H}$ are all satisfiable. Hence, the output graphs of the assignment tester $\mathcal{P}'$ also satisfy the property that all of its constraints are satisfiable if and only if the input variables encoded a satisfying a legal and satisfying assignment to the input constraint. 

We will define the constraint graph $G$ as $G' \circ \mathcal{P}'$. The high connectivity of each gadget $\overline{H}$ will be very useful to us in Step 3.

This process gives us a constraint graph $G$ with $\text{val}(G) \geq 1/8$, since a random strategy can achieve $\val(G) \geq 1/8$ in games over an alphabet of size 8.

\subsection{Step 3: Randomness-Efficient Parallel Repetition via Graph Powering}
\label{contrived-step-3}
 Fix a vertex $v$ in the game $G$. This vertex lies in $G_{(u',w')}$ for some $(u',w') \in E'$, where $G_{(u', w')}$ denotes the output of the assignment tester $\mathcal{P}'$ on $[u']$ and $[w']$. Now consider the graph $G^{*2}$. The label to $v$ in $G^{*2}$ claims labels to all vertices in $G_{(u',w')}$ due to the fact that $G_{(u',w')}$ is a complete graph. This label is valid if and only if all the constraints in $G_{(u',w')}$ are satisfied, which occurs if and only if the labels to $[u']$ and $[w']$ encode valid and satisfying labels for the edge $(u', w') \in E'$. Therefore, even picking a uniform vertex in $G^{*2}$ and testing the validity of its label already performs a uniform test in $G'$, and hence $\text{val}(G^{*2}) \leq \text{val}(G') < \varepsilon$. 

As discussed in Section \ref{gap_amplification}, the amount of randomness used to sample a random constraint in $G^{*2}$ consists of the randomness to query a single vertex of $G^{*2}$, which consists the randomness required to select a single vertex of $G$, and the randomness required to take a two step random walk in $G$. The degree of $G$ is a function of two things: the size of the output graphs of the assignment testers and the degree of $G'$. Both of these are constant in our setting, and so taking a two step walk on $G$ takes constant amount of randomness. Hence, using a derandomized parallel repetition scheme, we can transform a game $G$ with $\text{val}(G) \geq 1/8$ to a game $G^{*2}$ with $\text{val}(G^{*2}) < \varepsilon$ for an arbitrarily small constant $\varepsilon$, where the size of $G^{*2}$ is $|G^{*2}| = c(\Sigma, d)|G|$, and $\Sigma$ and $d$ denote the alphabet size and degree of the game $G'$. Since $\Sigma$ and $d$ are functions of $1/\eps$, for constant $\eps$ these are also constant.

We note that to get soundness $\varepsilon$ will normal parallel repetition, we would have had to repeat the game at least $k = \log_8 \frac{1}{\varepsilon}$ times, and so the size of this game $G^k$ would be $|G|^k = \Omega(|G|^{\log \frac{1}{\varepsilon}})$. We can see that $G^{*2}$ is considerably smaller than this, and is in fact almost-linear in $|G|$.

\section{Random games are fortified}
\label{sec:random_games_fortified}
In this section we prove that randomly sampled $d$-regular bipartite graphs are fortified with high probability, and can therefore be used as input games to the Main Theorem. Formally, we prove the following:
\begin{lemma}
\label{random_csp}
	Let $0 < \eta, \delta < 1$. Let $0 < \beta < 1/2$. Let $t$ be an integer and let $\Sigma$ be a finite alphabet. Let $d > \frac{4 (1 + \ln |\Sigma|)}{\eta^2 \delta^2}$. Let $\mathcal{G} = ([t] \times [t],E)$ be a bipartite graph that is the union of $d$ random perfect matchings $M^1,\ldots,M^d$, and let $G = (X,Y,E,\pi,\Sigma)$ be a game where $X = Y = [t]$ and for each edge $e \in E$, $\pi_e$ is a randomly chosen subset of $\Sigma \times \Sigma$ of density $\beta$. Then the following properties hold with probability at least $.99$:
	\begin{enumerate}
		\item $\mathcal{G}$ is $d$-regular, and has at most $200d^2$ parallel edges.
		\item For all $S,T \subseteq [t]$ with $|S|,|T| \geq \delta t$, we have
		$$
			\left | \frac{| E \cap (S \times T) |}{|S| |T|} - \frac{d}{t} \right | \leq  \eta \frac{d}{t}.	
		$$
		\item $\mathrm{val}(G) \leq \beta + \eta$.
		\item $G$ is $(\delta,2 \eta)$-fortified.
	\end{enumerate}
\end{lemma}

Note that, if we set $\eps = 2 \eta$ and assume that $200d^2 < \eps |E|$, the games provided by Lemma~\ref{random_csp} satisfy the conditions we require in Theorem~\ref{main_thm}. Before proving the lemma, we prove a general lemma about the sampling properties of $d$ random perfect bipartite matchings.

\newcommand{\M}{\mathcal{M}}


\begin{lemma}[Random matchings sample well]
\label{lemma_random_permutation}
	Let $M^1,\ldots,M^d$ be $d$ perfect matchings on $[t] \times [t]$ sampled uniformly at random. Let $Z\subseteq [t]\times[t]$ be an arbitrary set, and let $\mu = |Z|/t^2$. Then with probability at least $1 - \exp(-\Omega(\rho^2 \mu^2 d t))$, $\left| | \bigcup_j M^j \cap Z| - \mu d  t \right | \leq \rho\cdot \mu d t$. 
\end{lemma}

\begin{proof}
	We treat the selection of a random matching $M^j$ as a result of a random process where first, the edges of the complete bipartite graph $K_{t,t}$ are ordered randomly, and then the edges in $M^j \subset K_{t,t}$ are revealed one by one according to this random order. Let $E_i^j$ denote the $i$th revealed edge in $M^j$. Let $Y_i^j$ be the indicator variable for whether $E_i^j \in Z$. Let $Y=\sum_{j} \sum_i Y_i^j$. 
	Imagine a random process that first reveals all the edges of $M^1$ one at a time, then all the edges of $M^2$ one at a time, and so forth. Define a sequence of $td + 1$ random variables $X_0,X_1^1,\ldots,X_t^1,X_1^2,\ldots,X_t^2,\ldots,X_1^d,\ldots,X_t^d$, where $X_0 = \E [ Y ]$ and
	$$
		X_i^j = \E[ Y\mid E_{\leq (j,i)} ]
	$$
	where $E_{\leq (j,i)}$ denotes the sequence $E_1^1, \ldots, E_t^{j-1}, E_1^j, \ldots, E_i^j$, i.e., all the edges in matchings $M^1,\ldots,M^{j-1}$, and the first $i$ edges in matching $M^j$. By construction, the random variable sequence $\{X_i^j\}$ forms a Doob martingale with respect to the sequence $\{E_i^j\}$. We wish to apply Azuma's inequality to this to show that $Y$ is tightly concentrated about its mean, which is
\[ X_0 = \E[Y] = \sum_j \sum_i \E[Y_i] = \mu d t,\]
by linearity of expectation and the fact that the marginal distribution on each edge of $M^j$ is a uniformly random edge in $K_{t,t}$. In order to apply Azuma's inequality, we need to establish that $\max \{ |X_{i}^j - X_{i-1}^j|, | X_1^j - X_n^{j-1} | \} < c$ for some constant $c$. We argue that $c = 4$.

We introduce some notation that will be useful for us. Let $U^j_i$ denote the complete bipartite graph on all the vertices that haven't been ``paired'' up by the edges $E_1^j,\ldots,E_{i}^j$. In other words, it is the subgraph of $K_{t,t}$ where the edges $E_1^j,\ldots,E_i^j$, and all adjacent edges to them are removed. Let $\M^j_i$ denote the set of all perfect matchings on $U^j_i$. Note that, for all $i$, the matching $M^j$ is contained in $\M^j_i$. We will let $U^j_0$ denote $K_{t,t}$ and $\M^j_0$ to simply be the set of all perfect matchings on $K_{t,t}$. Finally, for all matchings (not necessarily perfect) $M$ of $K_{t,t}$, let $\alpha(M)$ denote $| M \cap Z|$.

Consider the difference $|X_1^j - X_n^{j-1}|$. Suppose that the edges in the sequence $E_{< (j,1)}$ -- i.e., all the edges in matchings $M^1,\ldots,M^{j-1}$ -- have been revealed. Then we have
\begin{align*}
	X^j_1 - X^{j-1}_t &= \E [ Y \mid E_{\leq (j,1)} ] - \E [ Y \mid  E_{< (j,1)} ] \\
				&= \sum_{j' \geq j} \E \left  [ \sum_i Y_i^{j'} \middle | E_{\leq (j,1)} \right ] - \E \left [ \sum_i Y_i^{j'} \middle |E_{< (j,1)} \right ] \\
				&= \E \left  [ \sum_i Y_i^j  \middle | E_{\leq (j,1)} \right ] - \E \left [ \sum_i Y_i^j \middle |E_{< (j,1)} \right ] \\
				&= \E \left  [ \sum_i Y_i^{j} \middle | E_1^j \right ] - \E \left [ \sum_i Y_i^{j} \right ]
\end{align*}
In the second line we used the linearity of expectation and the fact that, conditioned on $E_{< (j,1)}$, the random variables $Y_{i'}^{j'}$ are all fixed (i.e. revealing more edges from other matchings do not change their values) for all $i'$ and all $j' < j$. In the third line, we use that revealing an edge in matching $M^j$ does not affect the random variables $Y_i^{j'}$ for $j' > j$. We use the same reasoning in the fourth line; $Y^j_i$ is independent of the edges of $M^1,\ldots,M^{j-1}$.

Observe that, conditioned on $E^j_1$, we have that $M^j$ is a uniformly distributed matching in $\M^j_1$ adjoined with $E^j_1$ (since $\M^j_1$ technically contains submatchings). Without conditioning on $E^j_1$, $M^j$ is a uniformly distributed matching in $\M^j_0$. Thus we have the identities
$$
\E \left  [ \sum_i Y_i^{j} \middle | E_1^j \right ] = Y^j_1 + \left( \frac{1}{| \M^j_1 |} \sum_{N \in \M^j_1} \alpha(N) \right)
$$
and
$$
\E \left  [ \sum_i Y_i^{j} \right ] = \frac{1}{| \M^j_0 |} \sum_{M \in \M^j_0} \alpha(M).
$$
Define the mapping $B: \M^j_0 \to \M^j_1$  on matchings where, for all matchings $M \in \M^j_0$:
\begin{itemize}
	\item If $M$ contains $E^j_1$, then $B(M)$ is the submatching $M$ restricted to $U_1^j$.
	\item Else if $M$ contains $(a,d)$ and $(c,b)$ where $E_1^j = (a,b)$, then $B(M)$ is the submatching $M$ restricted to $U_1^j$ adjoined with $(c,d)$ (which was not in originally in $M$).
\end{itemize}
Fix an $M \in \M^j_0$. Suppose that $E^j_1 \in M$. Then $| \alpha(M) - \alpha(B(M)) | \leq 1$. Otherwise, $| \alpha(M) - \alpha(B(M)) | \leq 2$, because it could be that both $(a,d)$ and $(b,c)$ are in $Z$, and $(c,d)$ is not. 

Furthermore, observe that the map $B$ is onto, and for all $N \in \M^j_1$, the sizes of the preimages $B^{-1}(N) \subset \M^j_0$ are all the same. Then we have
\begin{align*}
\E \left  [ \sum_i Y_i^{j} \right ] = \frac{1}{| \M^j_1 |} \sum_{N \in \M^j_1} \frac{|\M^j_1|}{|\M^j_0|} \sum_{M \in B^{-1}(N)} \alpha(M) 
\end{align*}
so
\begin{align*}
	\left | X^j_1 - X^{j-1}_n \right | &= \left | \E \left  [ \sum_i Y_i^{j} \middle | E_1^j \right ] - \E \left [ \sum_i Y_i^{j} \right ] \right | \\
		&\leq | Y^j_1 | + \left | \frac{1}{| \M^j_1 |} \sum_{N \in \M^j_1} \left ( \alpha(N) - \frac{|\M^j_1|}{|\M^j_0|} \sum_{M \in B^{-1}(N)} \alpha(M) \right)  \right |  \\
		&\leq 1 + \frac{1}{| \M^j_1 |} \sum_{N \in \M^j_1} \frac{|\M^j_1|}{|\M^j_0|} \sum_{M \in B^{-1}(N)} \left | \alpha(B(M)) - \alpha(M) \right |  \\
		&\leq 3.
\end{align*}
The first inequality follows from triangle inequality, the second inequality follows from the fact that the number of $M \in B^{-1}(N)$ is equal to $|\M^j_0|/|\M^j_1|$, and the third inequality follows from our bound on the difference $\left | \alpha(B(M)) - \alpha(M) \right |$. 

Since this holds for every fixing of $E_{< (j,1)}$, this implies that $|X^j_1 - X^{j-1}_t| < 4$ with certainty. The same argument as above also implies that for all $i$, $|X^j_{i+1} - X^j_i | < 4$ with certainty. Hence, we can apply Azuma's inequality:
\[ \Pr ( |X^d_n - X_0| \geq \rho \cdot \mu d t) \leq 2\exp \left(-\frac{\rho^2 \mu^2 dt}{2c^2}\right). \]
We conclude the theorem by observing that $X^d_n$ is the number of edges in the union of the matchings $M^1,\ldots,M^d$ that fall within $Z$.
\end{proof}

%

\begin{lemma}
\label{few_parallel_edges}
Let $\mathcal{G} = ([t] \times [t],E)$ be a bipartite graph that is the union of $d$ random perfect matchings on $[t] \times [t]$. Then the probability that there are more than $200 d^2$ parallel edges in $E$ is less than $1/200$.
\end{lemma}
\begin{proof}
	Let $M^1,\ldots,M^d$ denote the matchings. For $1 \leq j,j' \leq d$, and $1 \leq i \leq n$, let $X_{j,j',i}$ denote the indicator variable that the $i$th left vertex gets matched to the same right vertex under matchings $M^j$ and $M^{j'}$. Note that $\E [X_{j,j',i} ] = 1/t$. Note that the number of parallel edges is at most $ \sum_{j, j'} \sum_i X_{j,j,i'}$, and thus the expected number of parallel edges is at most $d^2$. By Markov's inequality, the number of parallel edges is at most $200 d^2$ with probability at  least $1 - 1/200$.
\end{proof}

\begin{corollary}
\label{edge_density}
Let $0 < \delta, \rho < 1$, and let $d > 1/(\rho^2 \delta^2) + 2$. Let $\mathcal{G} = ([t] \times [t],E)$ be a bipartite graph that is the union of $d$ random perfect matchings on $[t] \times [t]$. Then with probability at least $1 - \exp(-\Omega( \rho^2 \delta^2 dt))$, for every $S, T \subseteq [t]$ where $|S|,|T| \geq \delta t$, we have that
$$
	\left | \frac{| E \cap (S \times T) |}{|S| |T|} - \frac{d}{t} \right | \leq \rho \frac{d}{t}.
$$
\end{corollary}
\begin{proof}
	This follows from Lemma~\ref{lemma_random_permutation} and union bounding over all $S, T \subseteq [t]$ such that $|S|,|T| \geq \delta t$ (of which there are at most $2^{2t}$).
\end{proof}

\noindent We now prove Lemma~\ref{random_csp}, which we restate here for completeness.

\begin{proof}[of Lemma~\ref{random_csp}]

By Lemma~\ref{few_parallel_edges} and Corollary~\ref{edge_density}, we have that with probability at least $199/200 - \exp(-\Omega(\rho^2 \delta^2 dt)) \geq 198/200$, the graph $\mathcal{G}$ is such that properties (1) and (2) of the lemma statement are satisfied. Call this event $H$.

We now argue that properties (3) and (4) are satisfied with high probability, conditioned on $H$. Define $m := t d$.
\begin{claim}
\label{random_makes_fortified}
Let $S \subseteq X$ and $T \subseteq Y$ be such that $|S|,|T| \geq \delta t$. The probability that there exist assignments $\psi_X: X \to \Sigma$ and $\psi_Y: Y \to \Sigma$ such that more than $2\eta$ fraction of the constraints $\pi_e$ such that $e \in E \cap (S \times T)$ are satisfied by $(\psi_X,\psi_Y)$, conditioned on $H$, is at most $\exp( - (\eta^2 \delta^2 d - 2 (\ln |\Sigma|) )t)$.
\end{claim}
\begin{proof}
	Fix $\psi_X : X \to \Sigma$ and $\psi_Y: Y \to \Sigma$. Let $E_{S \times T}$ denote $E \cap (S \times T)$. We have that $|E_{S \times T}| \geq \delta td/2$.  Given a fixed assignment, the probability a randomly chosen constraint $\pi_e$ for an edge $e \in E_{S \times T}$ is satisfied by the assignment is $\beta$. Thus the expected fraction of satisfied edges is $\beta |E_{S,T}|$. By Chernoff, the probability that more than $(\beta + \eta) |E_{S \times T}|$, or less than  $(\beta - \eta) |E_{S \times T}|$ edges are satisfied is at most $\exp(- 2 \eta^2 |E_{S \times T}|) \leq \exp(-\eta ^2 \delta^2 m)$ by our condition on the size of $E_{S \times T}$.
	
	Union bounding over all $|\Sigma|^{2t} = \exp( 2 (\ln |\Sigma| )t) $ possible assignments $(\psi_X,\psi_Y)$, we have that the probability that there exists an assignment such that more than $2\eta |E_{S \times T}|$ edges are satisfied is at most $\exp( - (\eta^2 \delta^2 m - 2 (\ln |\Sigma|) t))$.
\end{proof}

Let $J_{S,T}$ denote the event that for all assignments $(\psi_X,\psi_Y)$, no more than $\beta + \eta$ fraction of edges in $E \cap (S \times T)$ are satisfied by $(\psi_X,\psi_Y)$. Let $J$ denote the event that $J_{S,T}$ holds for all $S, T$ of size at least $\delta n$. By union bound, the probability that $J$ does not hold is at most
\begin{align*}
	2^{2t} \cdot \exp( - (\eta^2 \delta^2 m - 2 (\ln |\Sigma|) t)) &= \exp( - (\eta^2 \delta^2 m - 2(1 + \ln |\Sigma| ) t)).
\end{align*}
Since
$$
	d > \max \left \{ \frac{2}{\delta} \ln \frac{1}{\delta} , \frac{4 (1 + \ln |\Sigma|)}{\eta^2 \delta^2} \right \} 
$$
then the probability that $J$ and $H$ both do not hold is at most
$$
	\Pr(\neg H) + \Pr(\neg J | H) \leq .99.
$$
But if $J$ and $H$ both hold, this implies that for all $S, T$ of size at least $\delta t$, the fraction of satisfiable edges is at between $\beta - \eta$ and $\beta + \eta$. Thus this implies that $G$ is $(\delta, 2\eta)$-fortified.
\end{proof}

\end{document}